\documentclass[aps, pra, reprint, superscriptaddress]{revtex4-2}

\pdfoutput=1

\usepackage{header}

\graphicspath{{graphics/}}

\begin{document}

\bibliographystyle{apsrev4-2}

\title{Efficient Algorithms for Approximating Quantum Partition Functions at Low Temperature}

\author{Tyler Helmuth}
\email{tyler.helmuth@durham.ac.uk}
\homepage{http://www.tylerhelmuth.net}
\affiliation{Department of Mathematical Sciences, Durham University, Durham, DH1 3LE, United Kingdom}

\author{Ryan L. Mann}
\email{mail@ryanmann.org}
\homepage{http://www.ryanmann.org}
\affiliation{Centre for Quantum Computation and Communication Technology, Centre for Quantum Software and Information, School of Computer Science, Faculty of Engineering \& Information Technology, University of Technology Sydney, NSW 2007, Australia}
\affiliation{School of Mathematics, University of Bristol, Bristol, BS8 1UG, United Kingdom}

\begin{abstract}
    We establish an efficient approximation algorithm for the partition functions of a class of quantum spin systems at low temperature, which can be viewed as stable quantum perturbations of classical spin systems. Our algorithm is based on combining the contour representation of quantum spin systems of this type due to Borgs, Koteck\'y, and Ueltschi with the algorithmic framework developed by Helmuth, Perkins, and Regts, and Borgs et al.
\end{abstract}

\maketitle

\section{Introduction}
\label{section:Introduction}

The use of approximation algorithms for computing quantities of interest in statistical mechanical systems is an old subject. In this paper we are interested in the computational complexity of such algorithms for quantum spin systems, i.e., when do provably correct and efficient algorithms exist for quantum spin systems? For discrete classical statistical mechanics systems at high temperatures this is a well-studied question, and we have a relatively complete understanding for some models, e.g., the hard-core model~\cite{weitz2006counting, sly2010computational, sly2014counting, galanis2016inapproximability}, although important problems remain open~\cite{dyer2004relative}. Our understanding of approximation algorithms for quantum spin systems at high temperatures is less advanced, but has received a good deal of recent attention~\cite{bravyi2015monte, mann2019approximation, harrow2020classical, kuwahara2020clustering, crosson2020classical, mann2021efficient, galanis2022complexity}. 

Approximation algorithms at low temperatures are not nearly as well understood as at high temperatures. The main difficulty is caused by the existence of highly correlated phases at low temperatures. This difficulty is present in discrete classical statistical mechanics systems such as the Ising model and the hard-core model. While it has long been known that this is not a barrier for the classical Ising model~\cite{jerrum1993polynomial}, it is only relatively recently that this difficulty has been circumvented for the hard-core model. The main method used has been a recasting of low-temperature models into effective high-temperature models via polymer or contour representations~\cite{helmuth2020algorithmic, jenssen2019algorithms, liao2019counting, borgs2020efficient, carlson2020efficient}. However, there has also been success with other methods~\cite{barvinok2019weighted, huijben2023sampling}.

Our understanding of quantum systems at low temperatures is limited, in part due to the possibility of having spontaneously broken continuous symmetries: in such situations even truncated correlations can decay slowly. When this does not occur, however, an understanding has been developed in some cases via Peierls and Pirogov--Sinai methods~\cite{ginibre1969existence, kennedy1985long, borgs1996low, datta1996alow}. We shall be most interested in the latter methods, which apply more generally, e.g., for quantum perturbations of classical systems~\cite{borgs1996low, datta1996alow}. These methods yield control over low-temperature quantum spin systems in the region of the phase diagram where the quantum part of the Hamiltonian does not lead to new phenomena, and moreover, these methods can be used to study degeneracy-breaking quantum effects~\cite{datta1996alow, datta1996blow}.

In this paper we consider the development of approximate counting algorithms for low-temperature quantum spin systems. For the reasons discussed above, our focus shall be restricted to models where only discrete symmetries are broken. More precisely, we shall show how the version of Pirogov--Sinai theory developed for quantum spin systems in Ref.~\cite{borgs1996low} can be turned into an efficient approximate counting algorithm, following the broad strategy of Refs.~\cite{helmuth2020algorithmic, borgs2020efficient}. As for the investigation of the low-temperature phase diagram, our main motivation is to establish a rigorous classification of the computational phases of quantum spin systems. Our results show that the low-temperature phase of the quantum spin systems considered in this paper can be efficiently simulated in the sense of approximate counting using a classical algorithm. We now give an informal statement of our main result.

Our main result concerns stable quantum perturbations of classical spin systems. Deferring a full definition to Section~\ref{section:Preliminaries}, these spin systems can be informally introduced as quantum spin systems with Hamiltonians \mbox{$H=H_\Phi+\lambda H_\Psi$} where $H_\Phi$ is diagonal in a basis indexed by a classical spin system, $H_\Psi$ is local, and $\abs{\lambda}$ is small. The adjective stable implies that the set $\Xi$ of classical ground states (i.e., when $\lambda=0$) is also the set of ground states for small $\abs{\lambda}>0$. Let $Z^{g}_G(\beta,\lambda)$ be the partition function of such a model on a finite subgraph $G$ of $\mathbb{Z}^\nu$ with boundary condition $g\in\Xi$. Our main result may be stated informally as follows. For a formal statement, we refer to reader to Theorem~\ref{theorem:ApproximationAlgorithmPartitionFunction}.
\begin{theorem}[Informal]
    Let $G$ be a finite induced subgraph of $\mathbb{Z}^\nu$. There exists constants $\beta^\star$ and $\lambda^\star$, such that, for all $g\in\Xi$, all $\beta\geq\beta^\star$, and all $\abs{\lambda}\leq\lambda^\star$, there is an efficient algorithm for approximating the partition function $Z^{g}_G(\beta,\lambda)$.
\end{theorem}
As intimated previously, our algorithm is based on the observation of Ref.~\cite{helmuth2020algorithmic} that Pirogov--Sinai theory for low-temperature classical spin systems can be used to obtain efficient approximate counting algorithms, and on the generalisation of Pirogov--Sinai theory to quantum perturbations of classical spin systems~\cite{borgs1996low}. The main challenge we face arises from the non-commutative nature of quantum spin systems, which necessitates somewhat complicated contour model representations for applying Pirogov--Sinai theory~\cite{borgs1996low, datta1996alow}. Our main contribution is to show that a modification of the representation of Ref.~\cite{borgs1996low} allows for an implementation of the strategy used in Refs.~\cite{helmuth2020algorithmic, borgs2020efficient}. That is, we show that contours can be enumerated sufficiently quickly, and their weights can be approximated sufficiently quickly. We stress that it is far from being clear that this is possible, as the representation of Ref.~\cite{borgs1996low} is based on the Dyson series. Thus, while the representation of Ref.~\cite{borgs1996low} provides the upper bounds needed for applying Pirogov--Sinai theory in a straightforward way, these bounds are not sufficient for an efficient approximation algorithm, where one must be able to efficiently and accurately estimate weights. Our inclusion-exclusion approach for expansion allows us to overcome this difficulty. 

The arguments in Ref.~\cite{borgs1996low} apply in a greater generality than those we consider. In particular, they do not restrict their attention to the situation in which all ground states are stable, and they also consider expansions for correlation functions. We have restricted our focus to algorithms for computing the partition function to simplify the presentation as much as possible. In the stable setting, efficient classical simulation of correlation functions could be obtained with the methods of this paper. An extension of our methods to allow for non-stable ground states (i.e., the full low-temperature phase diagram) should also be feasible; such an extension has been achieved in the classical setting of the Potts model~\cite{borgs2012tight, borgs2020efficient}. We leave such an extension, which carries an increase in combinatorial complexity, to future work. Other interesting future directions include extensions to hypergraph interactions, and to analysing boundary conditions other than the pure boundary conditions treated in this paper.

This paper is structured as follows. We introduce the necessary preliminaries in \mbox{Section~\ref{section:Preliminaries}}. In \mbox{Section~\ref{section:ContourExpansion}}, We show how the partition function of quantum spin systems admits a contour expansion and consequently a cluster expansion. In \mbox{Section~\ref{section:ConvergenceOfTheClusterExpansion}}, we establish criteria for the absolute convergence of the cluster expansion. Then, in \mbox{Section~\ref{section:ApproximationAlgorithm}}, we use this framework to establish our approximation algorithm for the quantum partition function at low temperature. Finally, we conclude in \mbox{Section~\ref{section:ConclusionAndOutlook}} with some remarks and further open problems.

\section{Preliminaries}
\label{section:Preliminaries}

\subsection{Graph Theory}

Let \mbox{$G=(V,E)$} be a graph with vertex set $V$ and edge set $E$. We denote the open and closed neighbourhood of a vertex $v \in V$ by \mbox{$\mathcal{N}(v)=\{u\mid\{u,v\}\in E\}$} and \mbox{$\mathcal{N}[v]=\mathcal{N}(v)\cup\{v\}$}, respectively. More generally, we denote the open and closed neighbourhood of a subset of vertices $U \subseteq V$ by \mbox{$\mathcal{N}(U)\coloneqq\bigcup_{v \in U}\mathcal{N}(v){\setminus}U$} and \mbox{$\mathcal{N}[U]\coloneqq\bigcup_{v \in U}\mathcal{N}[v]$}, respectively. For a subset $U$ of $V$, the induced subgraph $G[U]$ is the subgraph of $G$ whose vertex set is $U$ and whose edge set consists of all edges in $G$ which have both endpoints in $U$. By a slight abuse of notation, we often identify a vertex subset with the subgraph it induces. Let $G$ be a finite induced subgraph of an infinite lattice graph $\mathbb{L}$ (e.g. $\mathbb{Z}^\nu$). The set $V(G)$ partitions $\mathbb{L}[V(\mathbb{L}){\setminus}V(G)]$ into connected components. We define the \emph{interior} of $G$, $\operatorname{int}G$, to be the union of all components with finite support and define the \emph{exterior} of $G$, $\operatorname{ext}G$, to be the unique component with infinite support. The \emph{boundary} of $G$, denoted by $\partial G$, is the set of edges in $E(\mathbb{L})$ with an endpoint in $V(G)$ and an endpoint in $V(\mathbb{L}){\setminus}V(G)$. The \emph{interior boundary} of $G$, denoted by $\partial^\text{in}G$, is the set of vertices in $V(G)$ with a neighbour in $V(\mathbb{L}){\setminus}V(G)$ and the \emph{exterior boundary} of $G$, denoted by $\partial^\text{ex}G$, is the set of vertices in $V(\mathbb{L}){\setminus}V(G)$ with a neighbour in $V(G)$.

Let $\mathbb{Z}^\nu$ denote the $\nu$-dimensional lattice graph and let $\mathbb{Z}^\nu_m$ denote the lattice graph $\mathbb{Z}^\nu\times(\mathbb{Z}/m\mathbb{Z})$. The graph $\mathbb{Z}^\nu\times(\mathbb{Z}/m\mathbb{Z})$ is periodic in the final coordinate. A sequence $(G_n)_{n\in\mathbb{N}}$ of finite induced subgraphs of $\mathbb{Z}^\nu$ converges to $\mathbb{Z}^\nu$ in the sense of van Hove if and only if
{
\begin{enumerate}
    \item $(V(G_n))_{n\in\mathbb{N}}$ is an increasing sequence of subsets,
    \item $\bigcup_{n\in\mathbb{N}}V(G_n) = \mathbb{Z}^\nu$,
    \item $\lim_{n\to\infty}\frac{\abs{\partial^\text{in}G_n}}{\abs{V(G_n)}} = 0$.
\end{enumerate}
}

\subsection{Quantum Spin Systems}

A \emph{quantum spin system} is formally modelled on the lattice graph $\mathbb{Z}^\nu$ with $\nu\geq2$. At each vertex $v$ of $\mathbb{Z}^\nu$ there is a $d$-dimensional Hilbert space $\mathcal{H}_v$ with $d<\infty$. The Hilbert space on $\mathbb{Z}^\nu$ is given by \mbox{$\mathcal{H}\coloneqq\bigotimes_{v \in V(\mathbb{Z}^\nu)}\mathcal{H}_v$}. In the sequel, we only consider systems on finite induced subgraphs of $\mathbb{Z}^\nu$. We require several assumptions on the quantum spin system, which can summarised as saying that we consider stable quantum perturbations of classical spin systems. We now describe these assumptions; for brevity in the sequel, we refer to a quantum spin system satisfying these assumptions as a \emph{stable quantum perturbation of a classical spin system}. We assume that the Hamiltonian is of the form \mbox{$H=H_\Phi+\lambda H_\Psi$}, where $H_\Phi$ is diagonal in the basis $\ket{s}=\bigotimes_{v \in V(\mathbb{Z}^\nu)}\ket{s_v}$ of a classical spin space $\Omega=[d]^{V(\mathbb{Z}^\nu)}$ and $H_\Psi$ is a local quantum perturbation. We are interested in the \emph{quantum partition function} $Z(\beta,\lambda)$ at inverse temperature $\beta>0$, defined by \mbox{$Z(\beta,\lambda)\coloneqq\Tr\left[e^{-\beta H}\right]$}.

We assume that the Hamiltonian $H_\Phi$ has translation-invariant interactions and is of the form \mbox{$\expval{H_\Phi}{s}=\sum_{v \in V(\mathbb{Z}^\nu)}\Phi_s(v)$}, where $\Phi_s(v)$ is a real number that depends on $s\in\Omega$ only via the spins $s_u$ for which $u\in\mathcal{N}[v]$. We assume that the Hamiltonian $H_\Phi$ has finitely many ground states $\Xi$ with ground state energy $e_0$. We say that a vertex $v \in V(\mathbb{Z}^\nu)$ is in the ground state $g$ if the spin $s_u$ coincides with the ground state $g$ for all \mbox{$u\in\mathcal{N}[v]$}. Otherwise, we say the vertex is excited. We further assume that \emph{Peierls' condition} holds, that is, there exists a constant \mbox{$\alpha_0>0$}, such that \mbox{$\Phi_s(v) \geq e_0+\alpha_0$} for all excited vertices $v \in V(\mathbb{Z}^\nu)$ of all configurations $s\in\Omega$. 

We assume that the Hamiltonian $H_\Psi$ is of the form \mbox{$H_\Psi=\sum_{e \in E(\mathbb{Z}^\nu)}\Psi(e)$}, where $\Psi(e)$ is a self-adjoint operator on \mbox{$\mathcal{H}_e\coloneqq\bigotimes_{v \in e}\mathcal{H}_v$}. We further assume that \mbox{$\norm{\Psi(e)}\leq1$} for every $e \in E(\mathbb{Z}^\nu)$, where $\norm{\;\cdot\;}$ denotes the operator norm. Note that this is always possible by a rescaling of $\lambda$.

Let \mbox{$G=(V,E)$} be a finite subgraph of $\mathbb{Z}^\nu$. The Hilbert space on $G$ is given by \mbox{$\mathcal{H}_G\coloneqq\bigotimes_{v \in V}\mathcal{H}_v$}. Let $H_G^g$, $H_{G,\Phi}^g$, and $H_{G,\Psi}^g$ denote the Hamiltonians $H_G$, $H_{G,\Phi}$, and $H_{G,\Psi}$ on $\mathcal{H}_G$ with boundary conditions $g\in\Xi$ on $\mathcal{N}(V)$.  This means that $H_G^g$ is the operator on $\mathcal{H}_G$ defined by the partial expectation value $\matrixel{s_{\mathcal{N}(V)}}{H_{\mathbb{Z}^\nu[\mathcal{N}[V]]}}{s_{\mathcal{N}(V)}}$, where the spin configuration $s_{\mathcal{N}(V)}$ coincides with the ground state $g$ on $\mathcal{N}(V)$, and similarly for $H_{G,\Phi}^g$ and $H_{G,\Psi}^g$. We define the partition function $Z_G^g(\beta,\lambda)$ on $\mathcal{H}_G$ with boundary conditions $g\in\Xi$ on $\mathcal{N}(V)$ by \mbox{$Z_G^g(\beta,\lambda)\coloneqq\Tr_{\mathcal{H}_G}\left[e^{-\beta H_G^g}\right]$}. We assume that all ground states $\Xi$ are \emph{stable} and define this notion subsequently (see \mbox{Section~\ref{section:ConvergenceOfTheClusterExpansion}}). 

For the remainder of the paper \mbox{$G=(V,E)$} will denote a finite induced subgraph of $\mathbb{Z}^\nu$.

\subsection{Approximation Schemes}

A \emph{fully polynomial-time approximation scheme} for a sequence of complex numbers $(z_n)_{n\in\mathbb{N}}$ is a deterministic algorithm that, for any $n$ and $\epsilon>0$, produces a complex number $\hat{z}_n$ such that \mbox{$\abs{z_n-\hat{z}_n}\leq\epsilon\abs{z_n}$} in time polynomial in $n$ and $1/\epsilon$.

\subsection{Abstract Polymer Models}

An \emph{abstract polymer model} is a triple $(\mathcal{C}, w, \sim)$, where $\mathcal{C}$ is a countable set whose elements are called \emph{polymers}, \mbox{$w:\mathcal{C}\to\mathbb{C}$} is a function that assigns to each polymer $\gamma\in\mathcal{C}$ a \emph{weight} $w_\gamma\in\mathbb{C}$, and $\sim$ is a \emph{symmetric compatibility relation} such that each polymer is incompatible with itself. A set of polymers is called \emph{admissible} if all the polymers in the set are pairwise compatible. Note that the empty set is admissible. Let $\mathcal{G}$ denote the collection of all admissible sets of polymers from $\mathcal{C}$. Then the abstract polymer partition function is defined by
\begin{equation}
    Z(\mathcal{C},w) \coloneqq \sum_{\Gamma\in\mathcal{G}}\prod_{\gamma\in\Gamma}w_\gamma. \notag
\end{equation}

\subsection{Abstract Cluster Expansion}

Let $\Gamma$ be a non-empty ordered tuple of polymers. The \emph{incompatibility graph} $H_\Gamma$ of $\Gamma$ is the graph with vertex set $\Gamma$ and edges between any two polymers if and only if they are incompatible. $\Gamma$ is called a \emph{cluster} if its incompatibility graph $H_\Gamma$ is connected. Let $\mathcal{G}_C$ denote the set of all clusters of polymers from $\mathcal{C}$. The \emph{abstract cluster expansion}~\cite{kotecky1986cluster, friedli2017statistical} is a formal power series for $\log{Z(\mathcal{C},w)}$ in the variables $w_\gamma$, defined by
\begin{equation}
    \log(Z(\mathcal{C},w)) \coloneqq \sum_{\Gamma\in\mathcal{G}_C}\varphi(H_\Gamma)\prod_{\gamma\in\Gamma}w_\gamma, \notag
\end{equation}
where $\varphi(H)$ denotes the \emph{Ursell function} of a graph $H$:
\begin{equation}
    \varphi(H) \coloneqq \frac{1}{\abs{V(H)}!}\sum_{\substack{E \subseteq E(H) \\ \text{spanning} \\ \text{connected}}}(-1)^{\abs{E}}. \notag
\end{equation}

\subsection{Contour Models}

We now introduce contour models which formalise the idea that spin systems have a geometric interpretation. For a detailed overview of contour models, we refer the reader to Ref.~\cite{friedli2017statistical}. A \emph{contour} is a pair \mbox{$\gamma=(\bar{\gamma},\text{lab}_\gamma)$}. The \emph{support} $\bar{\gamma}$ of $\gamma$ is a finite connected subset of $V(\mathbb{Z}^\nu_m)$. The \emph{labelling function} $\text{lab}_\gamma$ of $\gamma$ labels each edge in the boundary $\partial\bar{\gamma}$ with a ground state $g\in\Xi$ in such a way that $\text{lab}_\gamma$ is constant on the boundary of all the connected components of $V(\mathbb{Z}^\nu_m){\setminus}\bar{\gamma}$. Let $\operatorname{int}_g\bar{\gamma}$ denote the union of all connected components of $\operatorname{int}\bar{\gamma}$ with label $g$. The \emph{level} $l_\gamma$ of a contour $\gamma\in\mathcal{C}$ is defined inductively as follows. If \mbox{$\operatorname{int}\bar{\gamma}=\varnothing$}, then $l_\gamma=0$. Otherwise, \mbox{$l_\gamma=1+\max\{l_{\gamma'} \mid \bar{\gamma}'\subseteq\operatorname{int}\bar{\gamma}\}$}.

Two contours $\gamma$ and $\gamma'$ are \emph{compatible} if $\mathbb{Z}^\nu_m[\bar{\gamma}\cup\bar{\gamma}']$ is disconnected. A contour $\gamma$ is of \emph{type} $g$ if its exterior is labelled $g$. Let $\mathcal{G}^g$ denote the collection of all sets of pairwise compatible contours of type $g$. Let $\Gamma$ be a set of pairwise compatible contours. A contour $\gamma\in\Gamma$ is \emph{external} in $\Gamma$ if \mbox{$\bar{\gamma}\cap\operatorname{int}\bar{\gamma}'=\varnothing$} for all $\gamma'\in\Gamma$. Let $\mathcal{G}_\text{ext}^g$ denote the collection of all sets of pairwise compatible contours that are external and of type $g$. $\Gamma$ is \emph{matching and of type} $g$ if the labelling of the contours in $\Gamma$ is constant on the boundary of all the connected components of $V(\mathbb{Z}^\nu_m){\setminus}\bigcup_{\gamma\in\Gamma}\bar{\gamma}$ and all external contours in $\Gamma$ are of type $g$. Let $\mathcal{G}_\text{match}^g$ denote the collection of all sets of pairwise compatible contours that are matching and of type $g$. 

A \emph{contour model} is a pair $(\mathcal{C},w)$, where $\mathcal{C}$ is a countable set of contours and \mbox{$w:\mathcal{C}\to\mathbb{C}$} is a function that assigns to each contour $\gamma\in\mathcal{C}$ a complex number $w_\gamma$ called the \emph{weight} of the contour. Contour models have a natural partition function associated to them, defined by
\begin{equation}
    Z^g(\mathcal{C},w) \coloneqq \sum_{\Gamma\in\mathcal{G}_\text{match}^g}\prod_{\gamma\in\Gamma}w_\gamma. \notag
\end{equation}

\section{The Contour Expansion}
\label{section:ContourExpansion}

In this section we shall show how the partition function of a quantum spin system admits a contour model representation~\cite{borgs1996low, datta1996alow, ueltschi1998discontinuous}. Let us define \mbox{$\hat{\beta}=\hat{\beta}(m)\coloneqq\frac{\beta}{m}$}. The choice of the parameter $m$ is technical and specified in the sequel. We have the following lemma. 
\begin{lemma}[{restate=[name=restatement]QuantumContourExpansion}]
    \label{lemma:QuantumContourExpansion}
    The partition function $Z_G^g(\beta,\lambda)$ admits the following contour representation.
    \begin{equation}
        Z_G^g(\beta,\lambda) = e^{-\hat{\beta}e_0\abs{V(G)}}\sum_{\Gamma\in\mathcal{G}_\text{match}^g}\prod_{\gamma\in\Gamma}w_\gamma. \notag
    \end{equation}
\end{lemma}
We prove \mbox{Lemma~\ref{lemma:QuantumContourExpansion}} in \mbox{Appendix~\ref{section:QuantumContourExpansion}}. The details of the contour model may be obtained by examining the proof. We emphasise that the set of contours depends on the graph $G$, and the weights $w_\gamma$ depend on the parameters of the quantum spin system.

It is now standard to derive a second representation of the partition function, known as the \emph{external contour representation}, which does not require a matching condition on the contours. By summing over sets of pairwise compatible contours that are external and of type $g$ and iterating over their interiors, we obtain the following corollary.
\begin{corollary}
    The partition function $Z_G^g(\beta,\lambda)$ admits the following external contour representation.
    \begin{equation}
        Z_G^g(\beta,\lambda) = e^{-\hat{\beta}e_0\abs{V(G)}}\sum_{\Gamma\in\mathcal{G}_\text{ext}^g}\prod_{\gamma\in\Gamma}w_\gamma^\text{ext}, \notag
    \end{equation}
    where
    \begin{equation}
        w_\gamma^\text{ext} \coloneqq w_\gamma\prod_{g'\in\Xi}Z_{\text{int}_{g'}\bar{\gamma}}^{g'}(\beta,\lambda). \notag
    \end{equation}
\end{corollary}

Now we derive a third representation of the partition function, known as the \emph{polymer model representation}, which does not require an external condition on the contours. By summing over sets of pairwise compatible contours and iterating (see Ref.~\cite[Section 7.3]{friedli2017statistical}), we obtain the following corollary. 
\begin{corollary}
    The partition function $Z_G^g(\beta,\lambda)$ admits the following polymer model representation.
    \begin{equation}
        Z_G^g(\beta,\lambda) = e^{-\hat{\beta}e_0\abs{V(G)}}\sum_{\Gamma\in\mathcal{G}^g}\prod_{\gamma\in\Gamma}w_\gamma^g, \notag
    \end{equation}
    where
    \begin{equation}
        w_\gamma^g \coloneqq w_\gamma\prod_{g'\in\Xi}\frac{Z_{\text{int}_{g'}\bar{\gamma}}^{g'}(\beta,\lambda)}{Z_{\text{int}_{g'}\bar{\gamma}}^{g}(\beta,\lambda)}. \notag
    \end{equation}
\end{corollary}
Note that this is an abstract polymer partition function up to a multiplicative factor. As an immediate corollary, we obtain a cluster expansion for $\log(Z_G^g(\beta,\lambda))$.
\begin{corollary}
    The partition function $Z_G^g(\beta,\lambda)$ admits the following cluster expansion.
    \begin{equation}
        \log(Z_G^g(\beta,\lambda)) \coloneqq -\hat{\beta}e_0\abs{V(G)}+\sum_{\Gamma\in\mathcal{G}_C}\varphi(H_\Gamma)\prod_{\gamma\in\Gamma}w_\gamma^g. \notag
    \end{equation}
\end{corollary}

Our algorithm is based on computing the \emph{truncated cluster expansion} for $\log(Z_G^g(\beta,\lambda))$:
\begin{equation}
    T_n(Z_G^g(\beta,\lambda)) \coloneqq -\hat{\beta}e_0\abs{V(G)}+\sum_{\substack{\Gamma\in\mathcal{G}_C \\ \abs{\bar{\Gamma}} < n}}\varphi(H_\Gamma)\prod_{\gamma\in\Gamma}w_\gamma^g, \notag
\end{equation}
where \mbox{$\abs{\bar{\Gamma}}\coloneqq\sum_{\gamma\in\Gamma}\abs{\bar{\gamma}}$}.

\section{Convergence of the Cluster Expansion}
\label{section:ConvergenceOfTheClusterExpansion}

In this section we shall establish criteria for the absolute convergence of the cluster expansion for $\log(Z_G^g(\beta,\lambda))$. In particular, we shall show that the polymer weights $w_\gamma^g$ satisfy a bound of the form \mbox{$\abs{w_\gamma^g} \leq e^{-\mu^\star\abs{\bar{\gamma}}}$} for a sufficiently large constant $\mu^\star$, which is sufficient to ensure absolute convergence of the cluster expansion~\cite{kotecky1986cluster}. Borgs, Koteck\'y, and Ueltschi~\cite{borgs1996low} established such a bound and we follow their analysis in our setting. We first establish criteria for the weights $w_\gamma$ to satisfy such a bound.
\begin{lemma}[{restate=[name=restatement]QuantumContourWeightExponentialDecay}]
    \label{lemma:QuantumContourWeightExponentialDecay}
    Let $\hat{\beta}>0$, $\lambda\in\mathbb{C}$, and $\alpha\geq0$ be such that \mbox{$\hat{\beta}\abs{\lambda}\leq\frac{e^{-2(\alpha+1)}}{2\nu+1}$}. Then, for all contours \mbox{$\gamma=(\bar{\gamma},\text{lab}_\gamma)$},
    \begin{align}
        \abs{w_\gamma} &\leq \left(d\left(e^{-\frac{\alpha}{2\nu}}+e^{-\hat{\beta}\alpha_0}\right)\right)^{\abs{\bar{\gamma}}}. \notag
    \end{align}
\end{lemma}
We prove \mbox{Lemma~\ref{lemma:QuantumContourWeightExponentialDecay}} in \mbox{Appendix~\ref{section:QuantumContourWeightExponentialDecay}}. While the polymer weights $w_\gamma^g$ do not satisfy the desired bound in general, it is possible to obtain such a bound when the ground states are stable. To do this we follow Ref.~\cite{borgs1996low}, which adapts a method due to Borgs and Imbrie~\cite{borgs1989unified} that is based on 
Zahradn{\'i}k's truncation approach~\cite{zahradnik1984alternate}. A contour $\gamma$ of type $g$ is \emph{stable} if
\begin{equation}
    \abs{\frac{Z_{\text{int}_{g'}\bar{\gamma}}^{g'}(\beta,\lambda)}{Z_{\text{int}_{g'}\bar{\gamma}}^{g}(\beta,\lambda)}} \leq e^{4\abs{\partial\text{int}_{g'}\bar{\gamma}}}. \notag
\end{equation}
for all $g'\in\Xi$. Let $\mathcal{G}_\text{stab}^g$ denote the collection of all sets of pairwise compatible contours that are stable and of type $g$. We define the \emph{truncated partition function} $T_G^g(\beta,\lambda)$ by
\begin{equation}
    T_G^g(\beta,\lambda) \coloneqq e^{-\hat{\beta}e_0\abs{V(G)}}\sum_{\Gamma\in\mathcal{G}_\text{stab}^g}\prod_{\gamma\in\Gamma}w_\gamma^g. \notag
\end{equation}
If Peierls' condition holds and the weights $w_\gamma$ satisfy \mbox{$\abs{w_\gamma} \leq e^{-\mu\abs{\bar{\gamma}}}$} for $\mu$ sufficiently large, then the cluster expansion for $\log(T_G^g(\beta,\lambda))$ converges absolutely and the free energy $f_{\mathbb{Z}^\nu}^g(\beta,\lambda)$ in the infinite volume limit exists for each ground state $g\in\Xi$, i.e.,
\begin{equation}
    f_{\mathbb{Z}^\nu}^g(\beta,\lambda) \coloneqq -\frac{1}{\beta}\lim_{\scriptstyle{G\to\mathbb{Z}^\nu}}\frac{1}{\abs{V(G)}}\log\left(T_G^g(\beta,\lambda)\right), \notag
\end{equation}
where the limit is taken in the sense of van Hove. A ground state $g$ is stable if \mbox{$\Re(f_{\mathbb{Z}^\nu}^g(\beta,\lambda)) \leq \Re(f_{\mathbb{Z}^\nu}^{g'}(\beta,\lambda))$} for all $g'\in\Xi$. The following lemma is obtained in Ref.~\cite[Section 5]{borgs1996low} by applying the adaption of Ref.~\cite{borgs1989unified}.
\begin{lemma}
    \label{lemma:QuantumStableGroundStatesPolymerWeight}
    Suppose that the weights $w_\gamma$ satisfy \mbox{$\abs{w_\gamma} \leq e^{-\mu\abs{\bar{\gamma}}}$}. Then, there exists a constant \mbox{$\mu_0=\mu_0(\nu,d,\Xi,\alpha_0)$}, such that, for all $\mu>\mu_0$, all stable ground states $g\in\Xi$, and all contours $\gamma$,
    \begin{equation}
        \prod_{g'\in\Xi}\abs{\frac{Z_{\text{int}_{g'}\bar{\gamma}}^{g'}(\beta,\lambda)}{Z_{\text{int}_{g'}\bar{\gamma}}^{g}(\beta,\lambda)}} \leq e^{2(\nu+1)\abs{\bar{\gamma}}}. \notag
    \end{equation}
\end{lemma}

We now establish criteria for the absolute convergence of the cluster expansion for $\log(Z_G^g(\beta,\lambda))$.
\begin{lemma}[{restate=[name=restatement]ConvergenceClusterExpansion}]
    \label{lemma:ConvergenceClusterExpansion}
    There exists constants \mbox{$\beta^\star=\beta^\star(\nu,d,\Xi,\alpha_0)$} and \mbox{$\lambda^\star=\lambda^\star(\nu,d,\Xi,\alpha_0)$}, such that, for all $g\in\Xi$, all $\beta\geq\beta^\star$, and all $\abs{\lambda}\leq\lambda^\star$, the cluster expansion for $\log(Z_G^g(\beta,\lambda))$ converges absolutely, \mbox{$Z_G^g(\beta,\lambda)\neq0$}, and for $n\in\mathbb{Z}^+$,
    \begin{equation}
        \abs{T_n(Z_G^g(\beta,\lambda))-\log(Z_G^g(\beta,\lambda))} \leq \abs{V}e^{-\Omega(n)}. \notag
    \end{equation}
\end{lemma}
We prove \mbox{Lemma~\ref{lemma:ConvergenceClusterExpansion}} in \mbox{Appendix~\ref{section:ConvergenceClusterExpansion}}. This lemma implies that to obtain a multiplicative $\epsilon$-approximation to $Z_G^g(\beta,\lambda)$, it is sufficient to approximate the truncated cluster expansion $T_n(Z_G^g(\beta,\lambda))$ to order \mbox{$n=O(\log(\abs{V(G)}/\epsilon))$}.

\section{Approximation Algorithm}
\label{section:ApproximationAlgorithm}

In this section we shall establish our approximation algorithm for $Z_G^g(\beta,\lambda)$. Our algorithm is based on combining the cluster expansion for $\log(Z_G^g(\beta,\lambda))$ with the algorithmic framework of Helmuth, Perkins, and Regts~\cite{helmuth2020algorithmic} and Borgs et al.~\cite{borgs2020efficient}. In particular, we shall establish an efficient algorithm for approximating the truncated cluster expansion $T_n(Z_G^g(\beta,\lambda))$ to order \mbox{$n=O(\log(\abs{V(G)}/\epsilon))$}. We require the following lemmas.
\begin{lemma}
    \label{lemma:ListContoursByLevel}
    The contours of size at most $n$ can be listed and ordered by level in time $\exp(O(n))\cdot\abs{V(G)}^{O(1)}$.
\end{lemma}
\begin{proof}
    This is essentially Ref.~\cite[Proposition 3.12]{borgs2020efficient}; this reference concerns a torus, but the argument also works in our context of $\mathbb{Z}^\nu_m$.
\end{proof}

\begin{lemma}
    \label{lemma:ListClustersAlgorithm}
    The clusters of size at most $n$ can be listed in time \mbox{$\exp(O(n))\cdot\abs{V(G)}^{O(1)}$}.
\end{lemma}
\begin{proof}
    This is Ref.~\cite[Theorem 6]{helmuth2020algorithmic}.
\end{proof}

\begin{lemma}
    \label{lemma:UrsellFunctionAlgorithm}
    The Ursell function $\varphi(H)$ can be computed in time $\exp(O(\abs{V(H)}))$.
\end{lemma}
\begin{proof}
    This is a result of Ref.~\cite{bjorklund2008computing}; see Ref.~\cite[Lemma 5]{helmuth2020algorithmic}.
\end{proof}

\begin{lemma}[{restate=[name=restatement]ContourWeightsAlgorithm}]
    \label{lemma:ContourWeightsAlgorithm}
    The weight $w_\gamma$ of a contour \mbox{$\gamma=(\bar{\gamma},\text{lab}_\gamma)$} can be computed in time $\exp(O(\abs{\bar{\gamma}}))$.
\end{lemma}
We prove \mbox{Lemma~\ref{lemma:ContourWeightsAlgorithm}} in \mbox{Appendix~\ref{section:ContourWeightsAlgorithm}}.

\begin{lemma}
    \label{lemma:ApproximateTruncatedClusterExpansion}
    Let $\tilde{T}_n(Z_G^g(\beta,\lambda)$ denote the truncated cluster expansion $T_n(Z_G^g(\beta,\lambda))$ with the weights $w_\gamma^g$ replaced by multiplicative $\frac{\epsilon\abs{\operatorname{int}\bar{\gamma}}}{\abs{V(G)}}$-approximations. Then, there exists constants \mbox{$\beta^\star=\beta^\star(\nu,d,\Xi,\alpha_0)$} and \mbox{$\lambda^\star=\lambda^\star(\nu,d,\Xi,\alpha_0)$}, such that, for all $g\in\Xi$, all $\beta\geq\beta^\star$, and all $\abs{\lambda}\leq\lambda^\star$,
    \begin{equation}
        \abs{T_n(Z_G^g(\beta,\lambda))-\tilde{T}_n(Z_G^g(\beta,\lambda)} \leq \frac{\epsilon}{4}. \notag
    \end{equation}
\end{lemma}
\begin{proof}
    This follows from Ref.~\cite[Lemma 2.3]{borgs2020efficient} and Ref.~\cite[Lemma 3.9]{borgs2020efficient}.
\end{proof}

\begin{lemma}[{restate=[name=restatement]TruncatedQuantumClusterExpansionApproximationAlgorithm}]
    \label{lemma:TruncatedQuantumClusterExpansionApproximationAlgorithm}
    There exists constants \mbox{$\beta^\star=\beta^\star(\nu,d,\Xi,\alpha_0)$} and \mbox{$\lambda^\star=\lambda^\star(\nu,d,\Xi,\alpha_0)$}, such that, for all $g\in\Xi$, all $\beta\geq\beta^\star$, and all $\abs{\lambda}\leq\lambda^\star$, the truncated cluster expansion $T_n(Z_G^g(\beta,\lambda))$ can be approximated up to an additive $\epsilon$-error in time \mbox{$\exp(O(n))\cdot\abs{V(G)}^{O(1)}\cdot(1/\epsilon)^{O(1)}$}.
\end{lemma}
We prove \mbox{Lemma~\ref{lemma:TruncatedQuantumClusterExpansionApproximationAlgorithm}} in \mbox{Appendix~\ref{section:TruncatedQuantumClusterExpansionApproximationAlgorithm}}. We now establish our main result.
\begin{theorem}
    \label{theorem:ApproximationAlgorithmPartitionFunction}
    Let $G$ be a finite induced subgraph of $\mathbb{Z}^\nu$.  There exists constants \mbox{$\beta^\star=\beta^\star(\nu,d,\Xi,\alpha_0)$} and \mbox{$\lambda^\star=\lambda^\star(\nu,d,\Xi,\alpha_0)$}, such that, for all $g\in\Xi$, all $\beta\geq\beta^\star$, and all $\abs{\lambda}\leq\lambda^\star$, there is a fully polynomial-time approximation scheme for the partition function $Z_G^g(\beta,\lambda)$ of a stable quantum perturbation of a classical spin system.
\end{theorem}
\begin{proof}
    Combining \mbox{Lemma~\ref{lemma:ConvergenceClusterExpansion}} and \mbox{Lemma~\ref{lemma:TruncatedQuantumClusterExpansionApproximationAlgorithm}} gives a fully polynomial-time approximation scheme for the partition function $Z_G^g(\beta,\lambda)$ for all $g\in\Xi$, all $\beta\geq\beta^\star$, and all $\abs{\lambda}\leq\lambda^\star$.
\end{proof}

\section{Conclusion \& Outlook}
\label{section:ConclusionAndOutlook}

We have established a polynomial-time approximation algorithm for partition functions of a class of quantum spin systems at low temperature. This class can be viewed as stable quantum perturbations of classical spin systems. Our algorithm is based on combining the contour representation of quantum spin systems of this type due to Borgs, Koteck\'y, and Ueltschi~\cite{borgs1996low} with the algorithmic framework developed by Helmuth, Perkins, and Regts~\cite{helmuth2020algorithmic}, and Borgs et al.~\cite{borgs2020efficient}.

It would be interesting to extend these results to a more general class of quantum models, e.g., to bosonic and fermionic systems that are perturbations of classical spin systems with finitely many ground states. It would be extremely interesting to go further, and to obtain algorithms for low-temperature quantum systems with an infinite degeneracy of ground states, e.g., when the set of ground states possesses a continuous symmetry. Pirogov--Sinai theory cannot be applied in such a circumstance. However, efficient algorithms may still exist. For example, the arboreal gas (a model of interacting symplectic fermions) possesses highly degenerate ground states~\cite{bauerschmidt2021percolation}, but nonetheless an efficient algorithm exists at all temperatures~\cite{anari2021log}.

\section*{Acknowledgements}

RLM was supported by the QuantERA ERA-NET Cofund in Quantum Technologies implemented within the European Union's Horizon 2020 Programme (QuantAlgo project), EPSRC grants EP/L021005/1, EP/R043957/1, and EP/T001062/1, and the ARC Centre of Excellence for Quantum Computation and Communication Technology (CQC2T), project number CE170100012. No new data were created during this study.

\onecolumngrid

\appendix

\section{Proof of Lemma~\ref*{lemma:QuantumContourExpansion}}
\label{section:QuantumContourExpansion}

\QuantumContourExpansion*

\begin{proof}
    We begin by introducing the transfer matrix \mbox{$T \coloneqq e^{-\hat{\beta}H_G^g}$} with \mbox{$\hat{\beta}=\frac{\beta}{m}$}, and rewriting the partition function $Z_G^g(\beta,\lambda)$ as \mbox{$Z_G^g(\beta,\lambda)=\Tr_{\mathcal{H}_G}\left[T^m\right]$}. By the principle of inclusion-exclusion (see for example~\cite[Theorem 12.1]{graham1995handbook}),
    \begin{equation}
        T = \sum_{S \subseteq E}\bar{w}_S, \notag
    \end{equation}
    where, recalling that \mbox{$H^{g}_{G} = H^{g}_{G,\Phi}+\lambda H^{g}_{G,\Psi}$},
    \begin{equation}
        \bar{w}_S \coloneqq (-1)^{\abs{S}}\sum_{S' \subseteq S}(-1)^{\abs{S'}} e^{-\hat{\beta}\left(H_{G,\Phi}^g+\lambda\sum_{e \in S'}\Psi(e)\right)}. \notag
    \end{equation}
    We now resum to obtain
    \begin{equation}
        T = \sum_{U \subseteq V}T(U), \notag
    \end{equation}
    where, introducing the notation \mbox{$\text{supp}(S) = \{v\in V \mid \text{there exists $e \in S$ with $v \in e$}\}$}, $T(U)$ is given by
    \begin{equation}
        T(U) \coloneqq \sum_{\substack{S \subseteq E \\ \text{supp}(S)=U}}\bar{w}_S. \notag
    \end{equation}
    We introduce the Hamiltonian $H_{G,\Phi}^g(X)$ that is obtained from $H_{G,\Phi}^g$ by restricting to a subset $X$ of $V$ so that \mbox{$\expval{H_{G,\Phi}^g(X)}{s}=\sum_{v \in X}\Phi_s^g(v)$}, where $\Phi_s^g$ denotes $\Phi_s$ with the boundary conditions $g$ inherited from $H_{G,\Phi}^g$. We now introduce the operator $T_X(U)$ that is obtained from $T(U)$ by replacing the Hamiltonian $H_{G,\Phi}^g$ by $H_{G,\Phi}^g(X)$. That is, 
    \begin{equation}
        T_X(U) \coloneqq \sum_{\substack{S \subseteq E \\ \text{supp}(S)=U}}(-1)^{\abs{S}}\sum_{S' \subseteq S}(-1)^{\abs{S'}} e^{-\hat{\beta}\left(H_{G,\Phi}^g(X)+\lambda\sum_{e \in S'}\Psi(e)\right)}. \notag
    \end{equation}
    For a subset $U$ of $V$, let $G_U$ be the graph with vertex set $U$ and edges between any two vertices $u$ and $v$ if and only if \mbox{$\mathcal{N}[u]\cap\mathcal{N}[v]\neq\varnothing$}. Further let $\{U_i\}_{i=1}^{k(U)}$ denote the collection of all sets of vertices that correspond to a connected component of $G_U$. Note that if $e_i$ and $e_j$ are edges that intersect $U_i$ and $U_j$ respectively, then $e_i$ and $e_j$ have no vertices in common if $i \neq j$. Thus, we may write
    \begin{equation}
        T(U)=e^{-\hat{\beta}H_{G,\Phi}^g(V\setminus\mathcal{N}[U])}\prod_{i=1}^{k(U)}T_{\mathcal{N}[U_i]}(U_i). \notag
    \end{equation}
    For a subset $U$ of $V$, let $\Omega_U$ denote the set of all configurations $s\in\Omega$ on $U$. We now expand $T(U)$ to obtain
    \begin{equation}
        T(U) = \sum_{s\in\Omega_{V \setminus U}}e^{-\hat{\beta}\matrixel{s}{H_{G,\Phi}^g(V\setminus\mathcal{N}[U])}{s}}\ketbra{s}{s}\bigotimes_{i=1}^{k(U)}\matrixel{s_{V \setminus U_i}}{T_{\mathcal{N}[U_i]}(U_i)}{s_{V \setminus U_i}}. \notag
    \end{equation}
    Let $X$ denote the set of all pairs \mbox{$\chi=(U,s)$} of vertices $U \subseteq V$ and spin configurations $s\in\Omega_{V \setminus U}$. Inserting our formula for $T(U)$ into our formula for $T$ and using this definition, we obtain
    \begin{equation}
        T = \sum_{\chi \in X}\bar{w}_\chi \notag
    \end{equation}
        where 
    \begin{equation}
        \bar{w}_\chi \coloneqq e^{-\hat{\beta}\matrixel{s_{V \setminus U}}{H_{G,\Phi}^g(V\setminus\mathcal{N}[U])}{s_{V \setminus U}}}\ketbra{s_{V \setminus U}}{s_{V \setminus U}}\bigotimes_{i=1}^{k(U)}\matrixel{s_{V \setminus U_i}}{T_{\mathcal{N}[U_i]}(U_i)}{s_{V \setminus U_i}}. \notag
    \end{equation}
    Thus, we may now write the partition function as
    \begin{equation}
        Z_G^g(\beta,\lambda) = \sum_{\chi_1,\ldots,\chi_m \in X}w_{\chi_1,\ldots,\chi_m}, \notag
    \end{equation}
    where the \emph{weight} $w_{\chi_1,\ldots,\chi_m}$ of the sequence $\chi_1,\dots,\chi_m$ is defined by
    \begin{equation}
        w_{\chi_1,\ldots,\chi_m} \coloneqq \Tr_{\mathcal{H}_G}\left[\prod_{i=1}^m\bar{w}_{\chi_i}\right]. \notag
    \end{equation}
    We shall now consider configurations on the vertex set $V\times[m]$ with \emph{time slices} $(V_t)_{t=1}^m$ where $V_t$ is the set \mbox{$V_t\coloneqq\{(v,t) \mid v \in V\}$}. Since we consider periodic boundary conditions, we identify $t=m+1$ with $t=1$. Further, we now associate the configuration $\chi_t$ with time slice $V_t$ and write \mbox{$\chi_t=(U_t,s_{V_t \setminus U_t})$} where $U_t$ is a subset of $V_t$. For a given collection of configurations $(\chi_t)_{t=1}^m$, we assign a variable $\sigma_{(v,t)}$ to each vertex $(v,t)$ in $V\times[m]$ such that
    \begin{equation}
        \sigma_{(v,t)} = 
        \begin{cases} 
            s_{(v,t)} & \text{if $(v,t) \in V_t{\setminus}U_t$} \\
            \varnothing & \text{if $(v,t) \in U_t$}
        \end{cases}, \notag
    \end{equation}
    where $\varnothing$ serves to indicate the presence of a quantum excitation. We say that a vertex $(v,t)$ is in the ground state $g$ if the variable $\sigma_{(u,t)}$ coincides with the ground state $g$ for all \mbox{$(u,t)\in\mathcal{N}[(v,t)]$}. Otherwise, we say the vertex is excited. Note that two consecutive classical configurations with \mbox{$\sigma_{(v,t)}=s_{(v,t)}$} and \mbox{$\sigma_{(v,t+1)}=s_{(v,t+1)}$} have weight zero unless \mbox{$\sigma_{(v,t)}=\sigma_{(v,t+1)}$}. We now extract an overall factor of $e^{-\hat{\beta}e_0\abs{V(G)}}$. Finally, for a given configuration $\sigma$, we consider the union $\mathcal{V}$ of all excited vertices. Let $\sigma_\mathcal{V}$ denote the configuration of $\sigma$ on $\mathcal{V}$. Further let $\text{lab}_\mathcal{V}$ be a labelling function that assigns a label to each edge in the boundary $\partial\mathcal{V}$ such that \mbox{$\text{lab}_\mathcal{V}(e)=g'$} if $e$ is an edge for a vertex in $\mathcal{V}$ and a vertex in the ground state $g'$. We now define the weight $w_{\sigma_\mathcal{V},\text{lab}_\mathcal{V}}$ by
    \begin{equation}
        w_{\sigma_\mathcal{V},\text{lab}_\mathcal{V}} \coloneqq e^{\hat{\beta}e_0\abs{V(G)}}w_{\chi_1,\ldots,\chi_m}. \notag
    \end{equation}
    Note that the weight $w_{\sigma_\mathcal{V},\text{lab}_\mathcal{V}}$ depends only on the configuration $\sigma_\mathcal{V}$ and the labelling function $\text{lab}_\mathcal{V}$. Let $\{\mathcal{V}_i\}_{i=1}^{k(\mathcal{V})}$ denote the connected components of $\mathcal{V}$, then we have
    \begin{equation}
        w_{\sigma_\mathcal{V},\text{lab}_\mathcal{V}} = \prod_{i=1}^{k(\mathcal{V})}w_{\sigma_{\mathcal{V}_i},\text{lab}_{\mathcal{V}_i}}. \notag
    \end{equation}
    A configuration $\sigma$ corresponds to a set of pairwise compatible contours $\Gamma$ by taking $\{\mathcal{V}_i\}_{i=1}^{k(\mathcal{V})}$ as the supports of the contours and $\{\text{lab}_{\mathcal{V}_i}\}_{i=1}^{k(\mathcal{V})}$ as the labelling functions. Note that the set $\Gamma$ is matching and of type $g$.  For a contour \mbox{$\gamma=(\bar{\gamma},\text{lab}_\gamma)$}, let $\Omega_\gamma$ denote the set of all configurations $\sigma_{\bar{\gamma}}$ on $\bar{\gamma}$ that are compatible with the ground states determined by the labelling function $\text{lab}_\gamma$. Then we resum to obtain
    \begin{equation}
        Z_G^g(\beta,\lambda) = e^{-\hat{\beta}e_0\abs{V(G)}}\sum_{\Gamma\in\mathcal{G}_\text{match}^g}\prod_{\gamma\in\Gamma}w_\gamma, \notag
    \end{equation}
    where
    \begin{equation}
        w_\gamma \coloneqq \sum_{\sigma_{\bar{\gamma}}\in\Omega_\gamma}w_{\sigma_{\bar{\gamma}},\text{lab}_\gamma}. \notag
    \end{equation}
    This is the contour representation of the partition function $Z_G^g(\beta,\lambda)$, completing the proof.
\end{proof}

\section{Proof of Lemma~\ref*{lemma:QuantumContourWeightExponentialDecay}}
\label{section:QuantumContourWeightExponentialDecay}

Our proof of \mbox{Lemma~\ref{lemma:QuantumContourWeightExponentialDecay}} is based on the analysis of Borgs, Koteck\'y, and Ueltschi~\cite{borgs1996low}. We first require the following lemma bounding the operator $T_X(U)$. 
\begin{lemma}
    \label{lemma:ExponentialDecayTransferMatrix}
    Let \mbox{$G=(V,E)$} be a subgraph of the lattice graph on $\mathbb{Z}^\nu$ and let $U \subseteq V$ be a subset of vertices of $G$. Further let $\hat{\beta}>0$, $\lambda\in\mathbb{C}$, and $\alpha\geq0$ be such that \mbox{$\hat{\beta}\abs{\lambda}\leq\frac{e^{-2(\alpha+1)}}{2\nu+1}$}. Then, for any configurations $s\in\Omega_V$ and $\bar{s}\in\Omega_V$, 
    \begin{equation}
        \abs{\matrixel{s}{T_{\mathcal{N}[U]}(U)}{\bar{s}}} \leq e^{-\hat{\beta}e_0\abs{\mathcal{N}[U]}}e^{-\alpha\abs{U}}. \notag
    \end{equation}
\end{lemma}
\begin{proof}
    We have
    \begin{align}
        \abs{\matrixel{s}{T_{\mathcal{N}[U]}(U)}{\bar{s}}} &\leq \norm{T_{\mathcal{N}[U]}(U)} \notag \\
        &\leq \sum_{\substack{S \subseteq E \\ \text{supp}(S)=U}}\norm{(-1)^{\abs{S}}\sum_{S' \subseteq S}(-1)^{\abs{S'}} e^{-\hat{\beta}\left(H_{G,\Phi}^g(\mathcal{N}[U])+\lambda\sum_{e \in S'}\Psi(e)\right)}}. \notag
    \end{align}
    Let $P$ denote the set of all sequences of edges from $E$, the edge set of $G$. By resumming and applying the Duhamel expansion, it follows from the triangle inequality and the submultiplicativity of the norm (as in the proof of Ref.~\cite[Lemma 4.2]{borgs1996low}) that
    \begin{align}
        \abs{\matrixel{s}{T_{\mathcal{N}[U]}(U)}{\bar{s}}} &\leq e^{-\hat{\beta}e_0\abs{\mathcal{N}[U]}}\sum_{\substack{S \subseteq E \\ \text{supp}(S)=U}}\sum_{\substack{\rho \in P \\ \text{supp}(\rho)=S}}\frac{\left(\hat{\beta}\abs{\lambda}\right)^{\abs{\rho}}}{\abs{\rho}!}\prod_{e \in \rho}\norm{\Psi(e)} \notag \\
        &\leq e^{-\hat{\beta}e_0\abs{\mathcal{N}[U]}}\sum_{\substack{S \subseteq E \\ \text{supp}(S)=U}}\sum_{\substack{\rho \in P \\ \text{supp}(\rho)=S}}\frac{\left(\hat{\beta}\abs{\lambda}\right)^{\abs{\rho}}}{\abs{\rho}!}. \notag
    \end{align}
    Since the subgraph $G[U]$ has at most $\nu\abs{U}$ edges, there are at most $\binom{\nu\abs{U}}{n}$ subsets of $n$ edges whose support is $U$. Furthermore, for such a subset, there are precisely $\genfrac{\{}{\}}{0pt}{}{k}{n}n!$ sequences $\rho$ of length $\abs{\rho}=k$ that correspond to it. Here $\genfrac{\{}{\}}{0pt}{}{k}{n}$ denotes the Stirling number of the second kind. Finally, these subsets must contain at least $\left\lceil\frac{\abs{U}}{2}\right\rceil$ edges. Thus, we may write
    \begin{align}
        \abs{\matrixel{s}{T_{\mathcal{N}[U]}(U)}{\bar{s}}} &\leq e^{-\hat{\beta}e_0\abs{\mathcal{N}[U]}}\sum_{n=\left\lceil\frac{\abs{U}}{2}\right\rceil}^{\nu\abs{U}}\binom{\nu\abs{U}}{n}\sum_{k=n}^\infty\genfrac{\{}{\}}{0pt}{}{k}{n}\frac{n!}{k!}\left(\hat{\beta}\abs{\lambda}\right)^k \notag \\
        &\leq e^{-\hat{\beta}e_0\abs{\mathcal{N}[U]}}\sum_{n=\left\lceil\frac{\abs{U}}{2}\right\rceil}^{\nu\abs{U}}\binom{\nu\abs{U}}{n}\sum_{k=n}^\infty\binom{k}{n}\left(\hat{\beta}\abs{\lambda}\right)^k. \notag
    \end{align}
    By interchanging the summations over $n$ and $k$, we obtain
    \begin{align}
        \abs{\matrixel{s}{T_{\mathcal{N}[U]}(U)}{\bar{s}}} &\leq e^{-\hat{\beta}e_0\abs{\mathcal{N}[U]}}\sum_{k=\left\lceil\frac{\abs{U}}{2}\right\rceil}^\infty\left(\hat{\beta}\abs{\lambda}\right)^k\sum_{n=0}^{\nu\abs{U}}\binom{\nu\abs{U}}{n}\binom{k}{n} \notag \\
        &= e^{-\hat{\beta}e_0\abs{\mathcal{N}[U]}}\sum_{k=\left\lceil\frac{\abs{U}}{2}\right\rceil}^\infty\binom{\nu\abs{U}+k}{k}\left(\hat{\beta}\abs{\lambda}\right)^k \notag \\
        &\leq e^{-\hat{\beta}e_0\abs{\mathcal{N}[U]}}\sum_{k=\left\lceil\frac{\abs{U}}{2}\right\rceil}^\infty\left(e(2\nu+1)\hat{\beta}\abs{\lambda}\right)^k. \notag
    \end{align}
    By taking \mbox{$\hat{\beta}\abs{\lambda}\leq\frac{e^{-2(\alpha+1)}}{2\nu+1}$}, we have
    \begin{align}
        \abs{\matrixel{s}{T_{\mathcal{N}[U]}(U)}{\bar{s}}} &\leq e^{-\hat{\beta}e_0\abs{\mathcal{N}[U]}}\sum_{k=\left\lceil\frac{\abs{U}}{2}\right\rceil}^\infty e^{-(2\alpha+1)k} \notag \\
        &\leq e^{-\hat{\beta}e_0\abs{\mathcal{N}[U]}}e^{-\alpha\abs{U}}, \notag
    \end{align}
    completing the proof.
\end{proof}

We now prove \mbox{Lemma~\ref{lemma:QuantumContourWeightExponentialDecay}}.

\QuantumContourWeightExponentialDecay*

\begin{proof}
    We bound $\abs{w_\gamma}$ from above by considering a sum over all configurations \mbox{$s\in\Omega_{\bar{\gamma}}$} on the support $\bar{\gamma}$ of $\gamma$. This gives
    \begin{equation}
        \abs{w_\gamma} \leq e^{\hat{\beta}e_0\abs{\bar{\gamma}}}\sum_{\substack{U\subset\bar{\gamma} \\ \mathcal{N}[U]\subseteq\bar{\gamma}}}\sum_{s\in\Omega_{\bar{\gamma}}}e^{-\hat{\beta}\matrixel{s_{\bar{\gamma} \setminus U}}{H_{G,\Phi}^g(\bar{\gamma}\setminus\mathcal{N}[U])}{s_{\bar{\gamma} \setminus U}}}\abs{\matrixel{s}{T_{\mathcal{N}[U]}(U)}{s}}. \notag
    \end{equation}
    By using Peierls' condition and applying \mbox{Lemma~\ref{lemma:ExponentialDecayTransferMatrix}}, we obtain
    \begin{align}
        \abs{w_\gamma} &\leq \left(de^{\hat{\beta}e_0}\right)^{\abs{\bar{\gamma}}}\sum_{\substack{U\subset\bar{\gamma} \\ \mathcal{N}[U]\subseteq\bar{\gamma}}}e^{-\hat{\beta}(e_0+\alpha_0)\abs{\bar{\gamma}\setminus\mathcal{N}[U]}}e^{-\hat{\beta}e_0\abs{\mathcal{N}[U]}}e^{-\alpha\abs{U}} \notag \\
        &= \left(de^{-\hat{\beta}\alpha_0}\right)^{\abs{\bar{\gamma}}}\sum_{\substack{U\subset\bar{\gamma} \\ \mathcal{N}[U]\subseteq\bar{\gamma}}}e^{\hat{\beta}\alpha_0\abs{\mathcal{N}[U]}}e^{-\alpha\abs{U}}. \notag
    \end{align}
    Now, by using \mbox{$\abs{\mathcal{N}[U]}\leq2\nu\abs{U}$}, we have
    \begin{align}
        \abs{w_\gamma} &\leq \left(de^{-\hat{\beta}\alpha_0}\right)^{\abs{\bar{\gamma}}}\sum_{\substack{U\subset\bar{\gamma} \\ \mathcal{N}[U]\subseteq\bar{\gamma}}}e^{\left(\hat{\beta}\alpha_0-\frac{\alpha}{2\nu}\right)\abs{\mathcal{N}[U]}} \notag \\
        &\leq \left(de^{-\hat{\beta}\alpha_0}\left(e^{\hat{\beta}\alpha_0-\frac{\alpha}{2\nu}}+1\right)\right)^{\abs{\bar{\gamma}}} \notag \\
        &= \left(d\left(e^{-\frac{\alpha}{2\nu}}+e^{-\hat{\beta}\alpha_0}\right)\right)^{\abs{\bar{\gamma}}}, \notag
    \end{align}
    completing the proof.
\end{proof}

\section{Proof of Lemma~\ref*{lemma:ConvergenceClusterExpansion}}
\label{section:ConvergenceClusterExpansion}

\ConvergenceClusterExpansion*

\begin{proof}
    We proceed by showing that the polymer weights $w_\gamma^g$ satisfy a bound of the form \mbox{$\abs{w_\gamma^g} \leq e^{-\mu^\star\abs{\bar{\gamma}}}$} for a sufficiently large constant $\mu^\star$. Let \mbox{$\alpha=\alpha(\nu,d,\mu^\star,\alpha_0)$} and $\beta^\star(\nu,d,\mu^\star,\alpha_0)$ be such that
    \begin{equation}
        d\left(e^{-\frac{\alpha}{2\nu}}+e^{-\beta^\star\alpha_0}\right) \leq e^{-(\mu^\star+2(\nu+1))}, \notag
    \end{equation}
    where we recall that $\alpha_0$ is the constant in Peierls' condition. By choosing $m$ such that \mbox{$\hat{\beta}\in[\beta^\star,2\beta^\star)$}, it follows from \mbox{Lemma~\ref{lemma:QuantumContourWeightExponentialDecay}} that there is a $\lambda^\star=\lambda^\star(\beta^\star,\alpha,\nu)$ such that for $|\lambda|\leq\lambda^\star$,
    \begin{equation}
        \abs{w_\gamma} \leq e^{-(\mu^\star+2(\nu+1))\abs{\bar{\gamma}}}. \notag
    \end{equation}
    Now, by applying \mbox{Lemma~\ref{lemma:QuantumStableGroundStatesPolymerWeight}} with $\mu^\star\geq\mu_0(\nu,d,\Xi,\alpha_0)$ and using the assumption that the ground states are all stable, we have
    \begin{equation}
        \abs{w_\gamma^g} \leq e^{-\mu^\star\abs{\bar{\gamma}}}, \notag
    \end{equation}
    The proof then follows from Ref.~\cite[Lemma 2.1]{borgs2020efficient}, which states that for abstract polymer models where the polymers are connected induced subgraphs of a bounded-degree graph, compatibility is defined by vertex disjointness, and the polymer weights $w_\gamma^g$ satisfy a bound of the form \mbox{$\abs{w_\gamma^g} \leq e^{-\mu^\star\abs{\bar{\gamma}}}$} for a sufficiently large constant $\mu^\star$, then the cluster expansion for $\log(Z_G^g(\beta,\lambda))$ converges absolutely, \mbox{$Z_G^g(\beta,\lambda)\neq0$}, and for $n\in\mathbb{Z}^+$,
    \begin{equation}
        \abs{T_n(Z_G^g(\beta,\lambda))-\log(Z_G^g(\beta,\lambda))} \leq \abs{V}e^{-\Omega(n)}. \notag
    \end{equation}
    In the application of this lemma, we have considered polymers as induced subgraphs of the subdivided lattice.
\end{proof}

\section{Proof of Lemma~\ref*{lemma:ContourWeightsAlgorithm}}
\label{section:ContourWeightsAlgorithm}

\ContourWeightsAlgorithm*

\begin{proof}
    We proceed by describing an explicit representation of the weight $w_\gamma$ of a contour. We say that a spin configuration \mbox{$s\in\Omega_{\bar{\gamma}}$} on $\bar{\gamma}$ is \emph{admissible} if all vertices are excited. We denote the support $\bar{\gamma}$ of a contour $\gamma$ on time slice $t$ by $\bar{\gamma}_t$. For a subset $U$ of $\bar{\gamma}$, we denote the support of $U$ on time slice $t$ by $U_t$. Then, the weight $w_\gamma$ of a contour $\gamma$ may be written as
    \begin{equation}
        w_\gamma = e^{\hat{\beta}e_0\abs{\bar{\gamma}}}\sum_{\substack{U\subset\bar{\gamma} \\ \mathcal{N}[U]\subseteq\bar{\gamma}}}\sum_{\substack{s\in\Omega_{\bar{\gamma}} \\ \text{admissible}}}\prod_{t=1}^me^{-\hat{\beta}\matrixel{s_{\bar{\gamma}_t \setminus U_t}}{H_{G,\Phi}^g(\bar{\gamma}_t \setminus\mathcal{N}[U_t])}{s_{\bar{\gamma}_t\setminus U_t}}}\matrixel{s_{\bar{\gamma}_t}}{T_{\mathcal{N}[U_t]}(U_t)}{s_{\bar{\gamma}_{t+1}}}. \notag
    \end{equation}
    The first sum is over all subsets $U$ of $\bar{\gamma}$ such that \mbox{$\mathcal{N}[U]\subseteq\bar{\gamma}$}, of which there are at most $2^{\abs{\bar{\gamma}}}$. For each of these subsets $U$, we sum over all admissible spin configurations \mbox{$s\in\Omega_{\bar{\gamma}}$} on $\bar{\gamma}$, of which there are at most $d^{\abs{\bar{\gamma}}}$. The transfer matrix \mbox{$T_{\mathcal{N}[U_t]}(U_t)$} is given by
    \begin{equation}
        T_{\mathcal{N}[U_t]}(U_t) = \sum_{\substack{S_t \subseteq E \\ \text{supp}(S_t)=U_t}}(-1)^{\abs{S_t}}\sum_{S_t' \subseteq S_t}(-1)^{\abs{S_t'}}e^{-\hat{\beta}(H_{G,\Phi}^g(\mathcal{N}[U_t])+\lambda\sum_{e \in S_t'}\Psi(e))}. \notag
    \end{equation}
    The transfer matrix may be computed in time $\exp(O(\abs{\bar{\gamma}}))$ by summing over all subsets $S_t$ of $E$ whose support is $U_t$, of which there are at most $\exp(O(\abs{\bar{\gamma}}))$, and then summing over all subsets of $S_t$, of which there are at most $\exp(O(\abs{\bar{\gamma}}))$. It then follows that the weight $w_\gamma$ of a contour $\gamma$ can be computed in time $\exp(O(\abs{\bar{\gamma}}))$, completing the proof.
\end{proof}

\section{Proof of Lemma~\ref*{lemma:TruncatedQuantumClusterExpansionApproximationAlgorithm}}
\label{section:TruncatedQuantumClusterExpansionApproximationAlgorithm}

Our proof is based on the analysis of Helmuth, Perkins, and Regts~\cite{helmuth2020algorithmic}, and Borgs et al.~\cite{borgs2020efficient}.

\TruncatedQuantumClusterExpansionApproximationAlgorithm*

\begin{proof}
    We first list all contours of size at most $n$ and order this list by level. This list can be computed in time $\exp(O(n))\cdot\abs{V(G)}^{O(1)}$ by \mbox{Lemma~\ref{lemma:ListContoursByLevel}}. We now prove that, for a contour $\gamma$ of size at most $n$, we can approximate $w_\gamma^g$ up to a multiplicative $\frac{\epsilon\abs{\operatorname{int}\bar{\gamma}}}{\abs{V(G)}}$-error in time \mbox{$\exp(O(\abs{\bar{\gamma}}))\cdot\abs{V(G)}^{O(1)}\cdot(1/\epsilon)^{O(1)}$}. We shall prove this by induction on the level $l_\gamma$ of $\gamma$.
    
    The base case in the induction is a contour $\gamma$ with no interior. In this case, $w_\gamma^g$ can be computed exactly in time $\exp(O(\abs{\bar{\gamma}}))$ by \mbox{Lemma~\ref{lemma:ContourWeightsAlgorithm}}. Now suppose that the claim holds for all contours of level at most $t$. To approximate the weight $w_\gamma^g$ of a contour $\gamma$ at level $t+1$, we first compute $w_\gamma$, which can be achieved in time $\exp(O(\abs{\bar{\gamma}}))$ by \mbox{Lemma~\ref{lemma:ContourWeightsAlgorithm}}. We now approximate \mbox{$\prod_{g'\in\Xi}Z_{\text{int}_{g'}\bar{\gamma}}^{g'}(\beta,\lambda)/Z_{\text{int}_{g'}\bar{\gamma}}^{g}(\beta,\lambda)$} by computing the truncated cluster expansions for the logarithm of the partition functions to order \mbox{$n=O(\log(\abs{V(G)}/\epsilon))$} with the weights $w_\gamma^{g'}$ replaced by multiplicative $\frac{\epsilon\abs{\operatorname{int}\bar{\gamma}}}{\abs{V(G)}}$-approximations. Note that this only requires approximating the weights of contours at level at most $t$.
    
    The truncated cluster expansions can be computed to order $n$ in time \mbox{$\exp(O(n))\cdot\abs{V(G)}^{O(1)}\cdot(1/\epsilon)^{O(1)}$} as follows. We list all clusters of size at most $n$ in time \mbox{$\exp(O(n))\cdot\abs{V(G)}^{O(1)}$} by \mbox{Lemma~\ref{lemma:ListClustersAlgorithm}}. For each of these clusters, we compute the Ursell function in time $\exp(O(n))$ by \mbox{Lemma~\ref{lemma:UrsellFunctionAlgorithm}}, and approximate the weights up to a multiplicative $\frac{\epsilon\abs{\operatorname{int}\bar{\gamma}}}{\abs{V(G)}}$-error in time \mbox{$\exp(O(n))\cdot\abs{V(G)}^{O(1)}\cdot(1/\epsilon)^{O(1)}$} by the inductive hypothesis. Hence, the truncated cluster expansions can be computed to order \mbox{$n=O(\log(\abs{V(G)}/\epsilon))$} in time \mbox{$\abs{V(G)}^{O(1)}\cdot(1/\epsilon)^{O(1)}$}. It follows from \mbox{Lemma~\ref{lemma:ConvergenceClusterExpansion}} and \mbox{Lemma~\ref{lemma:ApproximateTruncatedClusterExpansion}} that this gives a multiplicative $\frac{\epsilon\abs{\operatorname{int}\bar{\gamma}}}{\abs{V(G)}}$-approximation to \mbox{$\prod_{g'\in\Xi}Z_{\text{int}_{g'}\bar{\gamma}}^{g'}(\beta,\lambda)/Z_{\text{int}_{g'}\bar{\gamma}}^{g}(\beta,\lambda)$}. Hence, we obtain a multiplicative $\frac{\epsilon\abs{\operatorname{int}\bar{\gamma}}}{\abs{V(G)}}$-approximation to $w_\gamma^g$ in time \mbox{$\exp(O(\abs{\bar\gamma}))\cdot\abs{V(G)}^{O(1)}\cdot(1/\epsilon)^{O(1)}$}, completing the induction.
    
    The truncated cluster expansion $T_n(Z_G^g(\beta,\lambda))$ can now be approximated by computing the truncated cluster expansion with the weights $w_\gamma^{g'}$ replaced by multiplicative $\frac{\epsilon\abs{\operatorname{int}\bar{\gamma}}}{\abs{V(G)}}$-approximations. By a similar argument, this can be computed in time \mbox{$\exp(O(n))\cdot\abs{V(G)}^{O(1)}\cdot(1/\epsilon)^{O(1)}$} and gives an additive $\epsilon$-approximation to $T_n(Z_G^g(\beta,\lambda))$, completing the proof.
\end{proof}

\twocolumngrid

\bibliography{bibliography}

\begin{thebibliography}{35}%
\makeatletter
\providecommand \@ifxundefined [1]{%
 \@ifx{#1\undefined}
}%
\providecommand \@ifnum [1]{%
 \ifnum #1\expandafter \@firstoftwo
 \else \expandafter \@secondoftwo
 \fi
}%
\providecommand \@ifx [1]{%
 \ifx #1\expandafter \@firstoftwo
 \else \expandafter \@secondoftwo
 \fi
}%
\providecommand \natexlab [1]{#1}%
\providecommand \enquote  [1]{``#1''}%
\providecommand \bibnamefont  [1]{#1}%
\providecommand \bibfnamefont [1]{#1}%
\providecommand \citenamefont [1]{#1}%
\providecommand \href@noop [0]{\@secondoftwo}%
\providecommand \href [0]{\begingroup \@sanitize@url \@href}%
\providecommand \@href[1]{\@@startlink{#1}\@@href}%
\providecommand \@@href[1]{\endgroup#1\@@endlink}%
\providecommand \@sanitize@url [0]{\catcode `\\12\catcode `\$12\catcode
  `\&12\catcode `\#12\catcode `\^12\catcode `\_12\catcode `\%12\relax}%
\providecommand \@@startlink[1]{}%
\providecommand \@@endlink[0]{}%
\providecommand \url  [0]{\begingroup\@sanitize@url \@url }%
\providecommand \@url [1]{\endgroup\@href {#1}{\urlprefix }}%
\providecommand \urlprefix  [0]{URL }%
\providecommand \Eprint [0]{\href }%
\providecommand \doibase [0]{https://doi.org/}%
\providecommand \selectlanguage [0]{\@gobble}%
\providecommand \bibinfo  [0]{\@secondoftwo}%
\providecommand \bibfield  [0]{\@secondoftwo}%
\providecommand \translation [1]{[#1]}%
\providecommand \BibitemOpen [0]{}%
\providecommand \bibitemStop [0]{}%
\providecommand \bibitemNoStop [0]{.\EOS\space}%
\providecommand \EOS [0]{\spacefactor3000\relax}%
\providecommand \BibitemShut  [1]{\csname bibitem#1\endcsname}%
\let\auto@bib@innerbib\@empty
\bibitem [{\citenamefont {Weitz}(2006)}]{weitz2006counting}%
  \BibitemOpen
  \bibfield  {author} {\bibinfo {author} {\bibfnamefont {D.}~\bibnamefont
  {Weitz}},\ }in\ \href {https://doi.org/10.1145/1132516.1132538} {\emph
  {\bibinfo {booktitle} {Proceedings of the Thirty-Eighth Annual ACM Symposium
  on Theory of Computing}}}\ (\bibinfo {organization} {ACM},\ \bibinfo {year}
  {2006})\ pp.\ \bibinfo {pages} {140--149}\BibitemShut {NoStop}%
\bibitem [{\citenamefont {Sly}(2010)}]{sly2010computational}%
  \BibitemOpen
  \bibfield  {author} {\bibinfo {author} {\bibfnamefont {A.}~\bibnamefont
  {Sly}},\ }in\ \href {https://doi.org/10.1109/FOCS.2010.34} {\emph {\bibinfo
  {booktitle} {2010 IEEE 51st Annual Symposium on Foundations of Computer
  Science}}}\ (\bibinfo {organization} {IEEE},\ \bibinfo {year} {2010})\ pp.\
  \bibinfo {pages} {287--296},\ \Eprint {https://arxiv.org/abs/1005.5584}
  {arXiv:1005.5584} \BibitemShut {NoStop}%
\bibitem [{\citenamefont {Sly}\ and\ \citenamefont
  {Sun}(2014)}]{sly2014counting}%
  \BibitemOpen
  \bibfield  {author} {\bibinfo {author} {\bibfnamefont {A.}~\bibnamefont
  {Sly}}\ and\ \bibinfo {author} {\bibfnamefont {N.}~\bibnamefont {Sun}},\
  }\href {https://doi.org/10.1214/13-AOP888} {\bibfield  {journal} {\bibinfo
  {journal} {The Annals of Probability}\ }\textbf {\bibinfo {volume} {42}},\
  \bibinfo {pages} {2383} (\bibinfo {year} {2014})}\BibitemShut {NoStop}%
\bibitem [{\citenamefont {Galanis}\ \emph {et~al.}(2016)\citenamefont
  {Galanis}, \citenamefont {{\v{S}}tefankovi{\v{c}}},\ and\ \citenamefont
  {Vigoda}}]{galanis2016inapproximability}%
  \BibitemOpen
  \bibfield  {author} {\bibinfo {author} {\bibfnamefont {A.}~\bibnamefont
  {Galanis}}, \bibinfo {author} {\bibfnamefont {D.}~\bibnamefont
  {{\v{S}}tefankovi{\v{c}}}},\ and\ \bibinfo {author} {\bibfnamefont
  {E.}~\bibnamefont {Vigoda}},\ }\href
  {https://doi.org/10.1017/s0963548315000401} {\bibfield  {journal} {\bibinfo
  {journal} {Combinatorics, Probability and Computing}\ }\textbf {\bibinfo
  {volume} {25}},\ \bibinfo {pages} {500} (\bibinfo {year} {2016})},\ \Eprint
  {https://arxiv.org/abs/1203.2226} {arXiv:1203.2226} \BibitemShut {NoStop}%
\bibitem [{\citenamefont {Dyer}\ \emph {et~al.}(2004)\citenamefont {Dyer},
  \citenamefont {Goldberg}, \citenamefont {Greenhill},\ and\ \citenamefont
  {Jerrum}}]{dyer2004relative}%
  \BibitemOpen
  \bibfield  {author} {\bibinfo {author} {\bibfnamefont {M.}~\bibnamefont
  {Dyer}}, \bibinfo {author} {\bibfnamefont {L.~A.}\ \bibnamefont {Goldberg}},
  \bibinfo {author} {\bibfnamefont {C.}~\bibnamefont {Greenhill}},\ and\
  \bibinfo {author} {\bibfnamefont {M.}~\bibnamefont {Jerrum}},\ }\href
  {https://doi.org/10.1007/s00453-003-1073-y} {\bibfield  {journal} {\bibinfo
  {journal} {Algorithmica}\ }\textbf {\bibinfo {volume} {38}},\ \bibinfo
  {pages} {471} (\bibinfo {year} {2004})}\BibitemShut {NoStop}%
\bibitem [{\citenamefont {Bravyi}(2015)}]{bravyi2015monte}%
  \BibitemOpen
  \bibfield  {author} {\bibinfo {author} {\bibfnamefont {S.}~\bibnamefont
  {Bravyi}},\ }\href@noop {} {\bibfield  {journal} {\bibinfo  {journal}
  {Quantum Information and Computation}\ }\textbf {\bibinfo {volume} {15}},\
  \bibinfo {pages} {1122} (\bibinfo {year} {2015})},\ \Eprint
  {https://arxiv.org/abs/1402.2295} {arXiv:1402.2295} \BibitemShut {NoStop}%
\bibitem [{\citenamefont {Mann}\ and\ \citenamefont
  {Bremner}(2019)}]{mann2019approximation}%
  \BibitemOpen
  \bibfield  {author} {\bibinfo {author} {\bibfnamefont {R.~L.}\ \bibnamefont
  {Mann}}\ and\ \bibinfo {author} {\bibfnamefont {M.~J.}\ \bibnamefont
  {Bremner}},\ }\href {https://doi.org/10.22331/q-2019-07-11-162} {\bibfield
  {journal} {\bibinfo  {journal} {Quantum}\ }\textbf {\bibinfo {volume} {3}},\
  \bibinfo {pages} {162} (\bibinfo {year} {2019})},\ \Eprint
  {https://arxiv.org/abs/1806.11282} {arXiv:1806.11282} \BibitemShut {NoStop}%
\bibitem [{\citenamefont {Harrow}\ \emph {et~al.}(2020)\citenamefont {Harrow},
  \citenamefont {Mehraban},\ and\ \citenamefont
  {Soleimanifar}}]{harrow2020classical}%
  \BibitemOpen
  \bibfield  {author} {\bibinfo {author} {\bibfnamefont {A.~W.}\ \bibnamefont
  {Harrow}}, \bibinfo {author} {\bibfnamefont {S.}~\bibnamefont {Mehraban}},\
  and\ \bibinfo {author} {\bibfnamefont {M.}~\bibnamefont {Soleimanifar}},\
  }in\ \href {https://doi.org/10.1145/3357713.3384322} {\emph {\bibinfo
  {booktitle} {Proceedings of the 52nd Annual ACM SIGACT Symposium on Theory of
  Computing}}}\ (\bibinfo {organization} {ACM},\ \bibinfo {year} {2020})\ pp.\
  \bibinfo {pages} {378--386},\ \Eprint {https://arxiv.org/abs/1910.09071}
  {arXiv:1910.09071} \BibitemShut {NoStop}%
\bibitem [{\citenamefont {Kuwahara}\ \emph {et~al.}(2020)\citenamefont
  {Kuwahara}, \citenamefont {Kato},\ and\ \citenamefont
  {Brand{\~a}o}}]{kuwahara2020clustering}%
  \BibitemOpen
  \bibfield  {author} {\bibinfo {author} {\bibfnamefont {T.}~\bibnamefont
  {Kuwahara}}, \bibinfo {author} {\bibfnamefont {K.}~\bibnamefont {Kato}},\
  and\ \bibinfo {author} {\bibfnamefont {F.~G.}\ \bibnamefont {Brand{\~a}o}},\
  }\href {https://doi.org/10.1103/physrevlett.124.220601} {\bibfield  {journal}
  {\bibinfo  {journal} {Physical Review Letters}\ }\textbf {\bibinfo {volume}
  {124}},\ \bibinfo {pages} {220601} (\bibinfo {year} {2020})},\ \Eprint
  {https://arxiv.org/abs/1910.09425} {arXiv:1910.09425} \BibitemShut {NoStop}%
\bibitem [{\citenamefont {Crosson}\ and\ \citenamefont
  {Slezak}(2020)}]{crosson2020classical}%
  \BibitemOpen
  \bibfield  {author} {\bibinfo {author} {\bibfnamefont {E.}~\bibnamefont
  {Crosson}}\ and\ \bibinfo {author} {\bibfnamefont {S.}~\bibnamefont
  {Slezak}},\ }\href@noop {} {\bibfield  {journal} {\bibinfo  {journal} {arXiv
  e-prints}\ } (\bibinfo {year} {2020})},\ \Eprint
  {https://arxiv.org/abs/2002.02232} {arXiv:2002.02232} \BibitemShut {NoStop}%
\bibitem [{\citenamefont {Mann}\ and\ \citenamefont
  {Helmuth}(2021)}]{mann2021efficient}%
  \BibitemOpen
  \bibfield  {author} {\bibinfo {author} {\bibfnamefont {R.~L.}\ \bibnamefont
  {Mann}}\ and\ \bibinfo {author} {\bibfnamefont {T.}~\bibnamefont {Helmuth}},\
  }\href {https://doi.org/10.1063/5.0013689} {\bibfield  {journal} {\bibinfo
  {journal} {Journal of Mathematical Physics}\ }\textbf {\bibinfo {volume}
  {62}},\ \bibinfo {pages} {022201} (\bibinfo {year} {2021})},\ \Eprint
  {https://arxiv.org/abs/2004.11568} {arXiv:2004.11568} \BibitemShut {NoStop}%
\bibitem [{\citenamefont {Galanis}\ \emph {et~al.}(2022)\citenamefont
  {Galanis}, \citenamefont {Goldberg},\ and\ \citenamefont
  {Herrera-Poyatos}}]{galanis2022complexity}%
  \BibitemOpen
  \bibfield  {author} {\bibinfo {author} {\bibfnamefont {A.}~\bibnamefont
  {Galanis}}, \bibinfo {author} {\bibfnamefont {L.~A.}\ \bibnamefont
  {Goldberg}},\ and\ \bibinfo {author} {\bibfnamefont {A.}~\bibnamefont
  {Herrera-Poyatos}},\ }\href {https://doi.org/10.1137/21M1454043} {\bibfield
  {journal} {\bibinfo  {journal} {SIAM Journal on Discrete Mathematics}\
  }\textbf {\bibinfo {volume} {36}},\ \bibinfo {pages} {2159} (\bibinfo {year}
  {2022})},\ \Eprint {https://arxiv.org/abs/2105.00287} {arXiv:2105.00287}
  \BibitemShut {NoStop}%
\bibitem [{\citenamefont {Jerrum}\ and\ \citenamefont
  {Sinclair}(1993)}]{jerrum1993polynomial}%
  \BibitemOpen
  \bibfield  {author} {\bibinfo {author} {\bibfnamefont {M.}~\bibnamefont
  {Jerrum}}\ and\ \bibinfo {author} {\bibfnamefont {A.}~\bibnamefont
  {Sinclair}},\ }\href {https://doi.org/10.1137/0222066} {\bibfield  {journal}
  {\bibinfo  {journal} {SIAM Journal on Computing}\ }\textbf {\bibinfo {volume}
  {22}},\ \bibinfo {pages} {1087} (\bibinfo {year} {1993})}\BibitemShut
  {NoStop}%
\bibitem [{\citenamefont {Helmuth}\ \emph {et~al.}(2020)\citenamefont
  {Helmuth}, \citenamefont {Perkins},\ and\ \citenamefont
  {Regts}}]{helmuth2020algorithmic}%
  \BibitemOpen
  \bibfield  {author} {\bibinfo {author} {\bibfnamefont {T.}~\bibnamefont
  {Helmuth}}, \bibinfo {author} {\bibfnamefont {W.}~\bibnamefont {Perkins}},\
  and\ \bibinfo {author} {\bibfnamefont {G.}~\bibnamefont {Regts}},\ }\href
  {https://doi.org/10.1007/s00440-019-00928-y} {\bibfield  {journal} {\bibinfo
  {journal} {Probability Theory and Related Fields}\ }\textbf {\bibinfo
  {volume} {176}},\ \bibinfo {pages} {851} (\bibinfo {year} {2020})},\ \Eprint
  {https://arxiv.org/abs/1806.11548} {arXiv:1806.11548} \BibitemShut {NoStop}%
\bibitem [{\citenamefont {Jenssen}\ \emph {et~al.}(2019)\citenamefont
  {Jenssen}, \citenamefont {Keevash},\ and\ \citenamefont
  {Perkins}}]{jenssen2019algorithms}%
  \BibitemOpen
  \bibfield  {author} {\bibinfo {author} {\bibfnamefont {M.}~\bibnamefont
  {Jenssen}}, \bibinfo {author} {\bibfnamefont {P.}~\bibnamefont {Keevash}},\
  and\ \bibinfo {author} {\bibfnamefont {W.}~\bibnamefont {Perkins}},\ }in\
  \href {https://doi.org/10.1137/1.9781611975482.135} {\emph {\bibinfo
  {booktitle} {Proceedings of the Thirtieth Annual ACM-SIAM Symposium on
  Discrete Algorithms}}}\ (\bibinfo {organization} {SIAM},\ \bibinfo {year}
  {2019})\ pp.\ \bibinfo {pages} {2235--2247},\ \Eprint
  {https://arxiv.org/abs/1807.04804} {arXiv:1807.04804} \BibitemShut {NoStop}%
\bibitem [{\citenamefont {Liao}\ \emph {et~al.}(2019)\citenamefont {Liao},
  \citenamefont {Lin}, \citenamefont {Lu},\ and\ \citenamefont
  {Mao}}]{liao2019counting}%
  \BibitemOpen
  \bibfield  {author} {\bibinfo {author} {\bibfnamefont {C.}~\bibnamefont
  {Liao}}, \bibinfo {author} {\bibfnamefont {J.}~\bibnamefont {Lin}}, \bibinfo
  {author} {\bibfnamefont {P.}~\bibnamefont {Lu}},\ and\ \bibinfo {author}
  {\bibfnamefont {Z.}~\bibnamefont {Mao}},\ }in\ \href
  {https://doi.org/10.4230/LIPIcs.APPROX-RANDOM.2019.34} {\emph {\bibinfo
  {booktitle} {Approximation, Randomization, and Combinatorial Optimization.
  Algorithms and Techniques (APPROX/RANDOM 2019)}}}\ (\bibinfo {organization}
  {Schloss Dagstuhl-Leibniz-Zentrum fuer Informatik},\ \bibinfo {year} {2019})\
  \Eprint {https://arxiv.org/abs/1903.07531} {arXiv:1903.07531} \BibitemShut
  {NoStop}%
\bibitem [{\citenamefont {Borgs}\ \emph {et~al.}(2020)\citenamefont {Borgs},
  \citenamefont {Chayes}, \citenamefont {Helmuth}, \citenamefont {Perkins},\
  and\ \citenamefont {Tetali}}]{borgs2020efficient}%
  \BibitemOpen
  \bibfield  {author} {\bibinfo {author} {\bibfnamefont {C.}~\bibnamefont
  {Borgs}}, \bibinfo {author} {\bibfnamefont {J.}~\bibnamefont {Chayes}},
  \bibinfo {author} {\bibfnamefont {T.}~\bibnamefont {Helmuth}}, \bibinfo
  {author} {\bibfnamefont {W.}~\bibnamefont {Perkins}},\ and\ \bibinfo {author}
  {\bibfnamefont {P.}~\bibnamefont {Tetali}},\ }in\ \href
  {https://doi.org/10.1145/3357713.3384271} {\emph {\bibinfo {booktitle}
  {Proceedings of the 52nd Annual ACM SIGACT Symposium on Theory of
  Computing}}}\ (\bibinfo {organization} {ACM},\ \bibinfo {year} {2020})\ pp.\
  \bibinfo {pages} {738--751},\ \Eprint {https://arxiv.org/abs/1909.09298}
  {arXiv:1909.09298} \BibitemShut {NoStop}%
\bibitem [{\citenamefont {Carlson}\ \emph {et~al.}(2020)\citenamefont
  {Carlson}, \citenamefont {Davies},\ and\ \citenamefont
  {Kolla}}]{carlson2020efficient}%
  \BibitemOpen
  \bibfield  {author} {\bibinfo {author} {\bibfnamefont {C.}~\bibnamefont
  {Carlson}}, \bibinfo {author} {\bibfnamefont {E.}~\bibnamefont {Davies}},\
  and\ \bibinfo {author} {\bibfnamefont {A.}~\bibnamefont {Kolla}},\
  }\href@noop {} {\bibfield  {journal} {\bibinfo  {journal} {arXiv e-prints}\ }
  (\bibinfo {year} {2020})},\ \Eprint {https://arxiv.org/abs/2003.01154}
  {arXiv:2003.01154} \BibitemShut {NoStop}%
\bibitem [{\citenamefont {Barvinok}\ and\ \citenamefont
  {Regts}(2019)}]{barvinok2019weighted}%
  \BibitemOpen
  \bibfield  {author} {\bibinfo {author} {\bibfnamefont {A.}~\bibnamefont
  {Barvinok}}\ and\ \bibinfo {author} {\bibfnamefont {G.}~\bibnamefont
  {Regts}},\ }\href {https://doi.org/10.1017/S0963548319000105} {\bibfield
  {journal} {\bibinfo  {journal} {Combinatorics, Probability and Computing}\
  }\textbf {\bibinfo {volume} {28}},\ \bibinfo {pages} {696} (\bibinfo {year}
  {2019})},\ \Eprint {https://arxiv.org/abs/1706.05423} {arXiv:1706.05423}
  \BibitemShut {NoStop}%
\bibitem [{\citenamefont {Huijben}\ \emph {et~al.}(2023)\citenamefont
  {Huijben}, \citenamefont {Patel},\ and\ \citenamefont
  {Regts}}]{huijben2023sampling}%
  \BibitemOpen
  \bibfield  {author} {\bibinfo {author} {\bibfnamefont {J.}~\bibnamefont
  {Huijben}}, \bibinfo {author} {\bibfnamefont {V.}~\bibnamefont {Patel}},\
  and\ \bibinfo {author} {\bibfnamefont {G.}~\bibnamefont {Regts}},\ }\href
  {https://doi.org/10.1002/rsa.21089} {\bibfield  {journal} {\bibinfo
  {journal} {Random Structures \& Algorithms}\ }\textbf {\bibinfo {volume}
  {62}},\ \bibinfo {pages} {219} (\bibinfo {year} {2023})},\ \Eprint
  {https://arxiv.org/abs/2103.07360} {arXiv:2103.07360} \BibitemShut {NoStop}%
\bibitem [{\citenamefont {Ginibre}(1969)}]{ginibre1969existence}%
  \BibitemOpen
  \bibfield  {author} {\bibinfo {author} {\bibfnamefont {J.}~\bibnamefont
  {Ginibre}},\ }\href {https://doi.org/10.1007/BF01645421} {\bibfield
  {journal} {\bibinfo  {journal} {Communications in Mathematical Physics}\
  }\textbf {\bibinfo {volume} {14}},\ \bibinfo {pages} {205} (\bibinfo {year}
  {1969})}\BibitemShut {NoStop}%
\bibitem [{\citenamefont {Kennedy}(1985)}]{kennedy1985long}%
  \BibitemOpen
  \bibfield  {author} {\bibinfo {author} {\bibfnamefont {T.}~\bibnamefont
  {Kennedy}},\ }\href {https://doi.org/10.1007/BF01206139} {\bibfield
  {journal} {\bibinfo  {journal} {Communications in Mathematical Physics}\
  }\textbf {\bibinfo {volume} {100}},\ \bibinfo {pages} {447} (\bibinfo {year}
  {1985})}\BibitemShut {NoStop}%
\bibitem [{\citenamefont {Borgs}\ \emph {et~al.}(1996)\citenamefont {Borgs},
  \citenamefont {Koteck{\'y}},\ and\ \citenamefont {Ueltschi}}]{borgs1996low}%
  \BibitemOpen
  \bibfield  {author} {\bibinfo {author} {\bibfnamefont {C.}~\bibnamefont
  {Borgs}}, \bibinfo {author} {\bibfnamefont {R.}~\bibnamefont {Koteck{\'y}}},\
  and\ \bibinfo {author} {\bibfnamefont {D.}~\bibnamefont {Ueltschi}},\ }\href
  {https://doi.org/10.1007/bf02101010} {\bibfield  {journal} {\bibinfo
  {journal} {Communications in Mathematical Physics}\ }\textbf {\bibinfo
  {volume} {181}},\ \bibinfo {pages} {409} (\bibinfo {year}
  {1996})}\BibitemShut {NoStop}%
\bibitem [{\citenamefont {Datta}\ \emph
  {et~al.}(1996{\natexlab{a}})\citenamefont {Datta}, \citenamefont
  {Fern{\'a}ndez},\ and\ \citenamefont {Fr{\"o}hlich}}]{datta1996alow}%
  \BibitemOpen
  \bibfield  {author} {\bibinfo {author} {\bibfnamefont {N.}~\bibnamefont
  {Datta}}, \bibinfo {author} {\bibfnamefont {R.}~\bibnamefont
  {Fern{\'a}ndez}},\ and\ \bibinfo {author} {\bibfnamefont {J.}~\bibnamefont
  {Fr{\"o}hlich}},\ }\href {https://doi.org/10.1007/bf02179651} {\bibfield
  {journal} {\bibinfo  {journal} {Journal of Statistical Physics}\ }\textbf
  {\bibinfo {volume} {84}},\ \bibinfo {pages} {455} (\bibinfo {year}
  {1996}{\natexlab{a}})}\BibitemShut {NoStop}%
\bibitem [{\citenamefont {Datta}\ \emph
  {et~al.}(1996{\natexlab{b}})\citenamefont {Datta}, \citenamefont
  {Fr{\"o}hlich}, \citenamefont {Rey-Bellet},\ and\ \citenamefont
  {Fern{\'a}ndez}}]{datta1996blow}%
  \BibitemOpen
  \bibfield  {author} {\bibinfo {author} {\bibfnamefont {N.}~\bibnamefont
  {Datta}}, \bibinfo {author} {\bibfnamefont {J.}~\bibnamefont {Fr{\"o}hlich}},
  \bibinfo {author} {\bibfnamefont {L.}~\bibnamefont {Rey-Bellet}},\ and\
  \bibinfo {author} {\bibfnamefont {R.}~\bibnamefont {Fern{\'a}ndez}},\ }\href
  {https://doi.org/10.5169/seals-116979} {\bibfield  {journal} {\bibinfo
  {journal} {Helvetica Physica Acta}\ }\textbf {\bibinfo {volume} {69}},\
  \bibinfo {pages} {752} (\bibinfo {year} {1996}{\natexlab{b}})}\BibitemShut
  {NoStop}%
\bibitem [{\citenamefont {Borgs}\ \emph {et~al.}(2012)\citenamefont {Borgs},
  \citenamefont {Chayes},\ and\ \citenamefont {Tetali}}]{borgs2012tight}%
  \BibitemOpen
  \bibfield  {author} {\bibinfo {author} {\bibfnamefont {C.}~\bibnamefont
  {Borgs}}, \bibinfo {author} {\bibfnamefont {J.~T.}\ \bibnamefont {Chayes}},\
  and\ \bibinfo {author} {\bibfnamefont {P.}~\bibnamefont {Tetali}},\ }\href
  {https://doi.org/10.1007/s00440-010-0329-0} {\bibfield  {journal} {\bibinfo
  {journal} {Probability Theory and Related Fields}\ }\textbf {\bibinfo
  {volume} {152}},\ \bibinfo {pages} {509} (\bibinfo {year} {2012})},\ \Eprint
  {https://arxiv.org/abs/1011.3058} {arXiv:1011.3058} \BibitemShut {NoStop}%
\bibitem [{\citenamefont {Koteck{\'y}}\ and\ \citenamefont
  {Preiss}(1986)}]{kotecky1986cluster}%
  \BibitemOpen
  \bibfield  {author} {\bibinfo {author} {\bibfnamefont {R.}~\bibnamefont
  {Koteck{\'y}}}\ and\ \bibinfo {author} {\bibfnamefont {D.}~\bibnamefont
  {Preiss}},\ }\href {https://doi.org/10.1007/bf01211762} {\bibfield  {journal}
  {\bibinfo  {journal} {Communications in Mathematical Physics}\ }\textbf
  {\bibinfo {volume} {103}},\ \bibinfo {pages} {491} (\bibinfo {year}
  {1986})}\BibitemShut {NoStop}%
\bibitem [{\citenamefont {Friedli}\ and\ \citenamefont
  {Velenik}(2017)}]{friedli2017statistical}%
  \BibitemOpen
  \bibfield  {author} {\bibinfo {author} {\bibfnamefont {S.}~\bibnamefont
  {Friedli}}\ and\ \bibinfo {author} {\bibfnamefont {Y.}~\bibnamefont
  {Velenik}},\ }\href {https://doi.org/10.1017/9781316882603} {\emph {\bibinfo
  {title} {Statistical Mechanics of Lattice Systems: A Concrete Mathematical
  Introduction}}}\ (\bibinfo  {publisher} {Cambridge University Press},\
  \bibinfo {year} {2017})\BibitemShut {NoStop}%
\bibitem [{\citenamefont {Ueltschi}(1998)}]{ueltschi1998discontinuous}%
  \BibitemOpen
  \bibfield  {author} {\bibinfo {author} {\bibfnamefont {D.}~\bibnamefont
  {Ueltschi}},\ }\emph {\bibinfo {title} {Discontinuous Phase Transitions in
  Quantum Lattice Systems}},\ \href@noop {} {Ph.D. thesis},\ \bibinfo  {school}
  {Verlag nicht ermittelbar} (\bibinfo {year} {1998})\BibitemShut {NoStop}%
\bibitem [{\citenamefont {Borgs}\ and\ \citenamefont
  {Imbrie}(1989)}]{borgs1989unified}%
  \BibitemOpen
  \bibfield  {author} {\bibinfo {author} {\bibfnamefont {C.}~\bibnamefont
  {Borgs}}\ and\ \bibinfo {author} {\bibfnamefont {J.~Z.}\ \bibnamefont
  {Imbrie}},\ }\href {https://doi.org/10.1007/BF01238860} {\bibfield  {journal}
  {\bibinfo  {journal} {Communications in Mathematical Physics}\ }\textbf
  {\bibinfo {volume} {123}},\ \bibinfo {pages} {305} (\bibinfo {year}
  {1989})}\BibitemShut {NoStop}%
\bibitem [{\citenamefont {Zahradn{\'i}k}(1984)}]{zahradnik1984alternate}%
  \BibitemOpen
  \bibfield  {author} {\bibinfo {author} {\bibfnamefont {M.}~\bibnamefont
  {Zahradn{\'i}k}},\ }\href {https://doi.org/10.1007/BF01212295} {\bibfield
  {journal} {\bibinfo  {journal} {Communications in Mathematical Physics}\
  }\textbf {\bibinfo {volume} {93}},\ \bibinfo {pages} {559} (\bibinfo {year}
  {1984})}\BibitemShut {NoStop}%
\bibitem [{\citenamefont {Bj{\"o}rklund}\ \emph {et~al.}(2008)\citenamefont
  {Bj{\"o}rklund}, \citenamefont {Husfeldt}, \citenamefont {Kaski},\ and\
  \citenamefont {Koivisto}}]{bjorklund2008computing}%
  \BibitemOpen
  \bibfield  {author} {\bibinfo {author} {\bibfnamefont {A.}~\bibnamefont
  {Bj{\"o}rklund}}, \bibinfo {author} {\bibfnamefont {T.}~\bibnamefont
  {Husfeldt}}, \bibinfo {author} {\bibfnamefont {P.}~\bibnamefont {Kaski}},\
  and\ \bibinfo {author} {\bibfnamefont {M.}~\bibnamefont {Koivisto}},\ }in\
  \href {https://doi.org/10.1109/FOCS.2008.40} {\emph {\bibinfo {booktitle}
  {49th Annual IEEE Symposium on Foundations of Computer Science}}}\ (\bibinfo
  {organization} {IEEE},\ \bibinfo {year} {2008})\ pp.\ \bibinfo {pages}
  {677--686},\ \Eprint {https://arxiv.org/abs/0711.2585} {arXiv:0711.2585}
  \BibitemShut {NoStop}%
\bibitem [{\citenamefont {Bauerschmidt}\ \emph {et~al.}(2021)\citenamefont
  {Bauerschmidt}, \citenamefont {Crawford},\ and\ \citenamefont
  {Helmuth}}]{bauerschmidt2021percolation}%
  \BibitemOpen
  \bibfield  {author} {\bibinfo {author} {\bibfnamefont {R.}~\bibnamefont
  {Bauerschmidt}}, \bibinfo {author} {\bibfnamefont {N.}~\bibnamefont
  {Crawford}},\ and\ \bibinfo {author} {\bibfnamefont {T.}~\bibnamefont
  {Helmuth}},\ }\href@noop {} {\bibfield  {journal} {\bibinfo  {journal} {arXiv
  e-prints}\ } (\bibinfo {year} {2021})},\ \Eprint
  {https://arxiv.org/abs/2107.01878} {arXiv:2107.01878} \BibitemShut {NoStop}%
\bibitem [{\citenamefont {Anari}\ \emph {et~al.}(2021)\citenamefont {Anari},
  \citenamefont {Liu}, \citenamefont {Gharan}, \citenamefont {Vinzant},\ and\
  \citenamefont {Vuong}}]{anari2021log}%
  \BibitemOpen
  \bibfield  {author} {\bibinfo {author} {\bibfnamefont {N.}~\bibnamefont
  {Anari}}, \bibinfo {author} {\bibfnamefont {K.}~\bibnamefont {Liu}}, \bibinfo
  {author} {\bibfnamefont {S.~O.}\ \bibnamefont {Gharan}}, \bibinfo {author}
  {\bibfnamefont {C.}~\bibnamefont {Vinzant}},\ and\ \bibinfo {author}
  {\bibfnamefont {T.-D.}\ \bibnamefont {Vuong}},\ }in\ \href
  {https://doi.org/10.1145/3406325.3451091} {\emph {\bibinfo {booktitle}
  {Proceedings of the 53rd Annual ACM SIGACT Symposium on Theory of
  Computing}}}\ (\bibinfo {organization} {ACM},\ \bibinfo {year} {2021})\ pp.\
  \bibinfo {pages} {408--420},\ \Eprint {https://arxiv.org/abs/2004.07220}
  {arXiv:2004.07220} \BibitemShut {NoStop}%
\bibitem [{\citenamefont {Graham}\ \emph {et~al.}(1995)\citenamefont {Graham},
  \citenamefont {Gr\"{o}tschel},\ and\ \citenamefont
  {Lov\'{a}sz}}]{graham1995handbook}%
  \BibitemOpen
  \bibfield  {author} {\bibinfo {author} {\bibfnamefont {R.~L.}\ \bibnamefont
  {Graham}}, \bibinfo {author} {\bibfnamefont {M.}~\bibnamefont
  {Gr\"{o}tschel}},\ and\ \bibinfo {author} {\bibfnamefont {L.}~\bibnamefont
  {Lov\'{a}sz}},\ }\href@noop {} {\emph {\bibinfo {title} {Handbook of
  Combinatorics}}},\ Vol.~\bibinfo {volume} {2}\ (\bibinfo  {publisher}
  {Elsevier},\ \bibinfo {year} {1995})\BibitemShut {NoStop}%
\end{thebibliography}%

\end{document}